\documentclass[journal,onecolumn]{IEEEtran}
\usepackage[normalem]{ulem}
\usepackage{url}
\usepackage{cite}

\ifCLASSINFOpdf

  \usepackage[pdftex]{graphicx}
  \DeclareGraphicsExtensions{.pdf,.jpg,.png}
  \graphicspath{C:/Users/Amichai/Desktop/BICA transactions/}
\else
\fi
\usepackage{caption}
\usepackage{amssymb}
\usepackage[cmex10]{amsmath}
\usepackage{enumerate}

\usepackage{mathtools}
\usepackage{array}
\usepackage{algorithm}
\usepackage{algorithmic}
\usepackage{xfrac }

\newtheorem{theorem}{Theorem}

\newtheorem{proposition}{Proposition}
\newtheorem{appendix_proposition}{Proposition}
\newtheorem{appendix_theorem}{Theorem}

\newenvironment{proof}[1][Proof]{\begin{trivlist}
\item[\hskip \labelsep {\bfseries #1}]}{\end{trivlist}}

\newcommand{\argmin}{\mathrm{arg}\displaystyle\min}
\newcommand{\argmax}{\mathrm{arg}\displaystyle\max}
\newcommand{\maxmax}{\displaystyle\max}

\newcommand{\qed}{\nobreak \ifvmode \relax \else
      \ifdim\lastskip<1.5em \hskip-\lastskip
      \hskip1.5em plus0em minus0.5em \fi \nobreak
      \vrule height0.75em width0.5em depth0.25em\fi}

\newcommand{\RNum}[1]{\uppercase\expandafter{\romannumeral #1\relax}}

\newcolumntype{M}[1]{>{\centering\arraybackslash}m{#1}}
\newcolumntype{N}{@{}m{0pt}@{}}

\DeclareFontFamily{OT1}{pzc}{}
\DeclareFontShape{OT1}{pzc}{m}{it}{<-> s * [1.10] pzcmi7t}{}
\DeclareMathAlphabet{\mathpzc}{OT1}{pzc}{m}{it}

\begin{document}
%
\title{Large Alphabet Source Coding using Independent Component Analysis}
%
%
%

\author{Amichai~Painsky,~\IEEEmembership{Member,~IEEE,}
        Saharon~Rosset and~Meir~Feder,~\IEEEmembership{Fellow,~IEEE}

\thanks{A. Painsky and S. Rosset are with the Statistics Department, Tel Aviv University, Tel Aviv, Israel. contact: amichaip@eng.tau.ac.il}
\thanks{M. Feder is with the Department of Electrical Engineering, Tel Aviv University, Tel Aviv, Israel}
\thanks{The material in this paper was presented in part at the International Symposium on Information Theory (ISIT) 2014 and the Data Compression Conference (DCC) 2015.}}

%
%

\markboth{IEEE TRANSACTIONS ON INFORMATION THEORY}%
{Shell \MakeLowercase{\textit{et al.}}: Bare Demo of IEEEtran.cls for Journals}
%



\maketitle

\begin{abstract}
Large alphabet source coding is a basic and well--studied problem in data compression. It has many applications such as compression of natural language text, speech and images. The classic perception of most commonly used methods is that a source is best described over an alphabet which is at least as large as the observed alphabet. In this work we challenge this approach and introduce a conceptual framework in which a large alphabet source is decomposed into ``as statistically independent as possible" components. This decomposition allows us to apply entropy encoding to each component separately, while benefiting from their reduced alphabet size.  We show that in many cases, such decomposition results in a sum of marginal entropies which is only slightly greater than the entropy of the source. Our suggested algorithm, based on a generalization of the Binary Independent Component Analysis, is applicable for a variety of large alphabet source coding setups. This includes the classical lossless compression, universal compression and high-dimensional vector quantization. In each of these setups, our suggested approach outperforms most commonly used methods. Moreover, our proposed framework is significantly easier to implement in most of these cases.      
\end{abstract}


%
\IEEEpeerreviewmaketitle

\section{Introduction}
%
%
%
%
\IEEEPARstart{A}{ssume} a source over an alphabet size $m$, from which a sequence of $n$ independent samples are drawn. The classical source coding problem is concerned with finding a sample-to-codeword mapping, such that the average codeword length is minimal, and the samples may be uniquely decodable. This problem was studied since the early days of information theory, and a variety of algorithms \cite{huffman1952method,witten1987arithmetic} and theoretical bounds \cite{cover2012elements} were introduced throughout the years.

The classical source coding problem usually assumes an alphabet size $m$ which is  small, compared with $n$. Here, we focus on a more difficult (and common) scenario, where the source's alphabet size is considered ``large" (for example, a word-wise compression of natural language texts).
In this setup, $m$ takes values which are either comparable (or even larger) than the length of the sequence $n$. 
The main challenge in large alphabet source coding is that the redundancy of the code, formally defined as the excess number of bits used over the source's entropy, typically increases with the alphabet size \cite{davisson1973universal}.

In this work we propose a conceptual framework for large alphabet source coding, in which we reduce the alphabet size by decomposing the source into multiple components which are ``as statistically independent as possible". This allows us to encode each of the components separately, while benefiting from the reduced redundancy of the smaller alphabet. 

To utilize this concept we introduce a framework based on a generalization of the  Binary Independent Component Analysis (BICA) method \cite{painsky2014generalized}. This framework efficiently searches for an invertible transformation  which minimizes the difference between the sum of marginal entropies (after the transformation is applied) and the joint entropy of the source. Hence, it minimizes the (attainable) lower bound on the average codeword length, when applying marginal entropy coding.

We show that while there exist sources which cannot be efficiently decomposed, their portion (of all possible sources over a given alphabet size) is small. Moreover,  we show that the difference between the sum of marginal entropies (after our transformation is applied) and the joint entropy, is bounded, on average, by a small constant for every $m$ (when the averaging takes place uniformly, over all possible sources of the same alphabet size $m$). This implies that our suggested approach is suitable for many sources of increasing alphabet size.

We demonstrate our method in a variety of large alphabet source coding setups. This includes the classical lossless coding, when the probability distribution of the source is known both to the encoder and the decoder, universal lossless coding, in which the decoder is not familiar with the distribution of the source, and lossy coding in the form of vector quantization. We show that our approach outperforms currently known methods in all these setups, for a variety of typical sources.  

The rest of this manuscript is organized as follows: After a short notation section, we review the work that was previously done in large alphabet source coding in Section \RNum{3}. Section \RNum{4} presents the generalized BICA problem, propose two different algorithms, and demonstrate their behavior on average and in worst-case. In Section \RNum{5} we apply our suggested framework to the classical lossless coding problem, over large alphabets. We then extend the discussion to universal compression in Section \RNum{6}. In Section \RNum{7} we further demonstrate our approach on vector quantization, with a special attention to high dimensional sources and low distortion. 
\section{Notation}

Throughout this paper we use the following standard notation: underlines denote vector quantities, where their respective components are written without underlines but with index. For example, the components of the $n$-dimensional vector $\underline{X}$ are $X_1, X_2, \dots X_n$.
Random variables are denoted with capital letters while their realizations are denoted with the respective lower-case letters. $P_{\underline{X}}\left(\underline{x}\right)  \triangleq P(X_1 = x_1, X_2= x_2\dots) $ is the probability function of  $\underline{X}$ while $H\left(\underline{X}\right)$ is the entropy of $\underline{X}$. This means  $H\left(\underline{X}\right)=-\sum_{\underline{x}} P_{\underline{X}}\left(\underline{x}\right) \log{P_{\underline{X}}\left(\underline{x}\right)}$ where the $\log{}$ function denotes a logarithm of base $2$ and $\lim_{x \to 0} x\log{(x)} = 0$. Specifically, we refer to the binary entropy of a Bernoulli distributed random variable $X \sim \operatorname{Ber}(p)$ as $H_b(X)$, while we denote the binary entropy function as $h_b (p)=-p\log{p}-(1-p)\log{(1-p)}$. 

\section{Previous Work}

In the classical lossless data compression framework, one usually assumes that both the encoder and the decoder are familiar with the probability distribution of the encoded source, $\underline{X}$. Therefore, encoding a sequence of $n$ memoryless samples drawn form this this source takes on average at least $n$ times its entropy $H\left(\underline{X}\right)$, for sufficiently large $n$ \cite{cover2012elements}. In other words, if $n$ is large enough to assume that the joint empirical entropy of the samples, $\hat{H}\left(\underline{X}\right)$, is close enough to the true joint entropy of the source, $H\left(\underline{X}\right)$, then $H\left(\underline{X}\right)$ is the minimal average number of bits required to encode a source symbol.
Moreover, it can be shown \cite{cover2012elements} that the minimum average codeword length,  $\bar{l}_{min}$,  for a uniquely decodable code, satisfies
\begin{equation}
\label{optimal_code}
H\left(\underline{X}\right) \leq \bar{l}_{min} \leq H\left(\underline{X}\right)+1.
\end{equation}

Entropy coding is a lossless data compression scheme that strives to achieve the lower bound, $\bar{l}_{min}=H\left(\underline{X}\right)$. Two of the most common entropy coding techniques are Huffman coding \cite{huffman1952method} and arithmetic coding \cite{witten1987arithmetic}.

The Huffman algorithm is an iterative construction of variable-length code table for encoding the source symbols. The algorithm derives this table from the probability of occurrence of each source symbol. Assuming these probabilities are dyadic (i.e.,
$-\log(p(\underline{x}))$ is an integer for every symbol $\underline{x} \in \underline{X}$), then the Huffman algorithm achieves $\bar{l}_{min}=H\left(\underline{X}\right)$. However, in the case where the probabilities are not dyadic then the Huffman code does not achieve the lower-bound of (\ref{optimal_code}) and may result in an average codeword length of up to  $H\left(\underline{X}\right)+1$ bits.
Moreover, although the Huffman code is theoretically easy to construct (linear in the number of symbols, assuming they are sorted according to their probabilities) it is practically a challenge to implement when the number of symbols increases \cite{moffat1997implementation}.  

Huffman codes achieve the minimum average codeword length among all uniquely decodable codes that assign a separate codeword to each symbol. However, if the probability of one of the symbols is close to $1$, a Huffman code with an average codeword length close to the entropy can only be constructed if a large number of symbols is jointly coded. The popular method of arithmetic coding is designed to overcome this problem. 

In arithmetic coding, instead of using a sequence of bits to represent a symbol, we represent it by a subinterval of the unit interval \cite{witten1987arithmetic}. This means that the code for a sequence of symbols is an interval whose length decreases as we add more symbols to the sequence. This property allows us to have a coding scheme that is incremental. In other words, the code for an extension to a sequence can be calculated simply from the code for the original sequence. Moreover, the codeword lengths are not restricted to be integral.
The arithmetic coding  procedure achieves an average length for the block that is within $2$ bits of the entropy. Although this is not necessarily optimal for any fixed block
length (as we show for Huffman code), the procedure is incremental and can be used for any block-length. Moreover, it does not requir the source probabilities to be dyadic. 
However, arithmetic codes are more complicated to implement and are a less likely to practically achieve the entropy of the source as the number of symbols increases. More specifically, due to the well-known underflow and overflow problems, finite precision implementations of the traditional adaptive arithmetic coding cannot work if the size of the source exceeds a certain limit \cite{yang2000universal}. For example, the widely used arithmetic coder by Witten et al. \cite{witten1987arithmetic} cannot work when the alphabet size is greater than $2^{15}$. The improved version of arithmetic coder by Moffat et al. \cite{moffat1998arithmetic} extends the alphabet to size $2^{30}$ by using low precision arithmetic, at the expense of compression performance.  

Notice that a large number of symbols not only results with difficulties in implementing entropy codes: as the alphabet size increases, we require an exponentially larger number of samples for the empirical entropy to converge to the true entropy. 

Therefore, when dealing with sources over large alphabets we usually turn to a universal compression framework. Here, we assume that the empirical probability distribution is not necessarily equal to the true distribution and henceforth unknown to the decoder. This means that a compressed representation of the samples now involves with two parts -- the compressed samples and an overhead redundancy (where the redundancy is defined as the number of bits used to transmit a message, minus the number of bits of actual information in the message). 

As mentioned above, encoding a sequence of $n$ samples, drawn from a memoryless source $\underline{X}$, requires at least $n$ times the empirical entropy, $\hat{H}(\underline{X})$. This is attained through entropy coding according to the source's empirical distribution. The redundancy, on the other hand, may be quantified in several ways. 

One common way of measuring the coding redundancy is through the minimax criterion \cite{davisson1973universal}. Here, the \textit{worst-case redundancy} is the lowest number of extra bits (over the empirical entropy) required in the worst case (that is, among all sequences) by any possible encoder.
Many worst-case redundancy results are known when the source's alphabet is finite. A succession of papers initiated by \cite{shtarkov1977coding} show that for the collection $\mathcal{I}_m^n$ of i.i.d. distributions over length-$n$ sequences drawn from an alphabet of a fixed size $m$, the worst-case redundancy behaves asymptotically as $\frac{m-1}{2}\log{\frac{n}{m}}$, as $n$ grows.

Orlitsky and Santhanam \cite{orlitsky2004speaking} extended this result to cases where $m$ varies with $n$. The \textit{standard compression} scheme they introduce differentiates between three situations in which $m=o(n)$, $n=o(m)$ and $m=\Theta(n)$. They provide leading term asymptotics and bounds to the worst-case minimax redundancy for these ranges of the alphabet size. Szpankowski and Weinberger \cite{szpankowski2012minimax} completed this study, providing the precise asymptotics to these ranges. For the purpose of our work we adopt the leading terms of their results, showing that the worst-case minimax redundancy, when $m \rightarrow \infty$, as $n$ grows, behaves as follows:
\begin{enumerate} [i]
\item For $m=o(n)$: \hfill \makebox[0pt][r]{%
            \begin{minipage}[b]{\textwidth}
              \begin{equation}
		\label{m=o(n)}
                 \quad \quad \quad \quad  \quad \quad \hat{R}(\mathcal{I}_m^n) \backsimeq \frac{m-1}{2}\log{\frac{n}{m}}+\frac{m}{2}\log{e}+\frac{m\log{e}}{3}\sqrt{\frac{m}{n}}  
              \end{equation}
          \end{minipage}}
\item For $n=o(m)$: \hfill \vspace{5pt}  \makebox[0pt][r]{%
            \begin{minipage}[b]{\textwidth}
              \begin{equation}
		\label{n=o(m)}
                  \; \quad  \quad \quad  \hat{R}(\mathcal{I}_m^n)  \backsimeq n\log{\frac{m}{n}}+\frac{3}{2}\frac{n^2}{m}\log{e}-\frac{3}{2}\frac{n}{m}\log{e}
              \end{equation}
          \end{minipage}}
\item $m=\alpha n+l(n)$: \hfill \makebox[0pt][r]{%
            \begin{minipage}[b]{\textwidth}
              \begin{equation}
		  \label{m=theta(n)}
               \; \; \;  \quad \quad \quad
 \hat{R}(\mathcal{I}_m^n)  \backsimeq n\log{B_\alpha}+l(n)\log{C_\alpha}-\log{\sqrt{A_\alpha}} 
              \end{equation}

          \end{minipage}}
\end{enumerate}
\hspace{26pt} where $\alpha$ is a positive constant, $l(n)=o(n)$ and
\begin{equation}
{\displaystyle C_\alpha \triangleq \frac{1}{2}+\frac{1}{2}\sqrt{1+\frac{4}{\alpha}}\quad , \quad A_\alpha \triangleq C_\alpha+\frac{2}{\alpha} \quad , \quad  B_\alpha \triangleq \alpha C_\alpha^{\alpha+2} e^{-\frac{1}{C_\alpha}}}. \nonumber 
\end{equation}

In a landmark paper from 2004, Orlitsky et al. \cite{orlitsky2004universal} presented a novel framework for universal compression of memoryless sources over unknown and possibly infinite alphabets.  
According to their framework, the description of any string, over any alphabet, can be viewed as consisting of two parts: the symbols appearing in the
string and the pattern that they form. For example, the string ``abracadabra" can be described by conveying the \textit{pattern} ``12314151231"
and the \textit{dictionary}

\begin{table}[!ht]
\centering
{
\begin{tabular}{|c|c|c|c|c|c|}
\hline
index &1 &2 &3 &4&5 \\
\hline
letter &a &b &r &c&d \\
\hline
\end{tabular}}

\end{table} 
\noindent Together, the pattern and dictionary specify that the string ``abracadabra" consists of the first letter to appear (a), followed by the second letter to appear (b), then by the third to appear (r), the first that appeared (a again), the fourth (c), etc.
Therefore, a compressed string involves with a compression of the pattern and its corresponding dictionary. Orlitsky et al. derived the bounds for pattern compression, showing that the redundancy of patterns compression under  i.i.d. distributions over potentially infinite alphabets is bounded by $\left( \frac{3}{2}\log{e} \right)n^{1/3}$. Therefore, assuming the alphabet size is $m$ and the number of uniquely observed symbols is $n_0$, the dictionary can be described in $n_0\log{m}$ bits, leading to an overall lower bound of $n_0\log{m}+n^{1/3}$ bits on the compression redundancy. 

An additional (and very common) universal compression scheme is the canonical Huffman coding \cite{witten1999managing}.  A canonical Huffman code is a particular type of Huffman code with unique properties which allow it to be described in a very compact manner. The advantage of a canonical Huffman tree is that one can encode a codebook in fewer bits than a fully described tree. Since a canonical Huffman codebook can be stored especially efficiently, most compressors start by generating a non-canonical Huffman codebook, and then convert it to a canonical form before using it. 

In canonical Huffman coding the bit lengths of each symbol are the same as in the traditional Huffman code. However, each code word is replaced with new code words (of the same length), such that a subsequent symbol is assigned the next binary number in sequence. 
For example, assume a Huffman code for four symbols, A to D:
\begin{table}[!ht]
\centering
{
\begin{tabular}{|c|c|c|c|c|}
\hline
symbol &A &B&C&D \\
\hline
codeword &11 &0 &101 &100 \\
\hline
\end{tabular}}
\end{table} 

Applying canonical Huffman coding to it we have
\begin{table}[!ht]
\centering
{
\begin{tabular}{|c|c|c|c|c|}
\hline
symbol &B &A&C&D \\
\hline
codeword &0 &10 &110 &111 \\
\hline
\end{tabular}}
\end{table} 

This way we do not need to store the entire Huffman mapping but only a list of all symbols in increasing order by their bit-lengths and record the number of symbols for each bit-length. This allows a more compact representation of the code, hence, lower redundancy.

An additional class of data encoding methods which we refer to in this work is lossy compression. In the lossy compression setup one applies inexact approximations for representing the content that has been encoded. In this work we focus on vector quantization, in which a high-dimensional vector $\underline{X} \in \mathbb{R}^d$ is to be represented by a finite number of points. Vector quantization works by clustering the observed samples of the vector $\underline{X}$ into groups, where each group is represented by its centroid point, such as in $k$-means and other clustering algorithms. Then, the centroid points that represent the observed samples are compressed in a lossless manner. 

In the lossy compression setup, one is usually interested in minimizing the amount of bits which represent the data for a given a distortion measure (or equivalently, minimizing the distortion for a given compressed data size).  The rate-distortion function defines the lower bound on this objective. It is defined as

\begin{equation}
R\left(D\right)=\min_{P(\underline{Y}|\underline{X})}I(\underline{X};\underline{Y})\,\, s.t. \,\, \mathbb{E} \left\{D(\underline{X},\underline{Y}) \right\} \leq D
\end{equation}

where $\underline{X}$ is the source, $\underline{Y}$ is recovered version of $\underline{X}$ and $D(\underline{X},\underline{Y})$ is some distortion measure between $\underline{X}$ and $\underline{Y}$. Notice that since the quantization is a deterministic mapping between $\underline{X}$ and $\underline{Y}$, we have that $I(\underline{X};\underline{Y})=H(\underline{Y})$. 

The Entropy Constrained Vector Quantization (ECVQ) is an iterative method for clustering the observed samples from $\underline{X}$ into centroid points which are later represented by a minimal average codeword length. The ECVQ algorithm minimizes the Lagrangian   

\begin{equation}
\label{ECVQ_intro}
L =\mathbb{E}\left\{D(\underline{X},\underline{Y})\right\}+\lambda\mathbb{E}\left\{l(\underline{X})\right\}
\end{equation}
where $\lambda$ is the Lagrange multiplier and $\mathbb{E}\left(l(\underline{X})\right)$ is the average codeword length for each symbol in $\underline{X}$. The ECVQ algorithm performs an iterative local minimization method similar to the generalized Lloyd algorithm \cite{lloyd1982least}. This means that for a given clustering of samples it constructs an entropy code to minimize the average codeword lengths of the centroids. Then, for a given coding of centroids it clusters the observed samples such that the average distortion is minimized, biased by the length of the codeword. This process continues until a local convergence occurs. 

The ECVQ algorithm performs local optimization (as a variant of the $k$-means algorithm) which is also not very scalable for an increasing number of samples. This means that in the presence of a large number of samples, or when the alphabet size of the samples is large enough, the clustering phase of the ECVQ becomes impractical. Therefore, in these cases, one usually uses a predefined lattice quantizer and only constructs a corresponding codebook for its centroids.

It is quite evident that large alphabet sources entails a variety of difficulties in all the compression setups mentioned above: it is more complicated to construct an entropy code for, it results in a great redundancy when universally encoded and it is much more challenging to design a vector quantizer for. In the following sections we introduce a framework which is intended to overcome these drawbacks.

\section{Generalized Binary Independent Component Analysis}

A common implicit assumption to most compression schemes in that the source is best represented over its observed alphabet size.
We would like to challenge this assumption, suggesting that in some cases there exists a transformation which decomposes a source into multiple ``as independent as possible" components whose alphabet size is much smaller.

\subsection{Problem Formulation}

Suppose we are given a binary random vector $\underline{X} \sim \underline{p}$ of a dimension $d$. We are interested in an invertible transformation $\underline{Y}=g(\underline{X})$ such that \underline{Y} is of the same dimension and alphabet size, $g:2^d \rightarrow 2^d$. In addition we would like the components (bits) of \underline{Y} to be as ``statistically independent as possible".
Notice that an invertible transformation of a vector \underline{X} is actually a one-to-one mapping (i.e. permutation) of its $m=2^d$ alphabet symbols. Therefore, there exist $2^d!$ possible invertible transformations. 

To quantify the statistical independence among the components of the vector \underline{Y} we use the well-known \textit{total correlation} measure as a multivariate generalization of the mutual information, 
\begin{equation}\label{eq:min_criterion}
C(\underline{Y})={\displaystyle \sum_{j=1}^{d}{H_b(Y_j)}-H(\underline{Y})}.
\end{equation}

This measure can also be viewed as the cost of encoding the vector $\underline{Y}$ component-wise, as if its components were statistically independent, compared to its true entropy. Notice that the total correlation is non-negative and equals zero iff the components of $\underline{Y}$ are mutually independent. Therefore, ``as statistically independent as possible" may be quantified by minimizing $C(\underline{Y})$. The total correlation measure was first considered as an objective for minimal redundancy representation by Barlow \cite{barlow1989finding}. It is also not new to finite field ICA problems, as demonstrated in \cite{attux2011immune}.

Since we define \underline{Y} to be an invertible transformation of \underline{X} we have $H(\underline{Y})=H(\underline{X})$ and our minimization objective is
\begin{equation}
{\displaystyle \sum_{j=1}^{d}{H_b(Y_j)} \rightarrow min.}
\label{eq:sum_ent_min_binary}
\end{equation}

We notice that $P(Y_j=0)$ is the sum of probabilities of all words whose $j^{th}$ bit equals $0$. We further notice that the optimal transformation is not unique. For example, we can always invert the  $j^{th}$ bit of all words, or even shuffle the bits, to achieve the same minimum.

In the following sections we review and introduce several methods for solving (\ref{eq:sum_ent_min_binary}). 
As a first step towards this goal we briefly describe the generalized BICA method. A complete derivation of this framework appears in \cite{painsky2015generalized}. Then, Sections \ref{ordering solution},\ref{worst case} and \ref{average case}  provide a simplified novel method for (\ref{eq:sum_ent_min_binary}) and discuss its theoretical properties.

\subsection{Piece-wise Linear Relaxation Algorithm}
\label{piece-wise solution}
In this section we briefly review our suggested method, as it appears in detail in \cite{painsky2015generalized}.

Let us first notice that the problem we are dealing with (\ref{eq:sum_ent_min_binary}) is a concave minimization problem over a discrete permutation set which is a hard problem. However, let us assume for the moment that instead of our ``true" objective (\ref{eq:sum_ent_min_binary}) we have a simpler linear objective function. That is, 

\begin{equation} \label{eq:linear}
{\displaystyle L(\underline{Y})=\sum_{j=1}^{d}{a_j \pi_j+b_j}=\sum_{i=1}^{m}{c_i P(\underline{Y}=y(i))}+d}
\end{equation}
where $\pi_j=p(Y_j=0)$ and the last equality changes the summation over $d$ bits to a summation over all $m=2^d$ symbols.

In order to minimize this objective function over the $m$ given probabilities $\underline{p}$ we simply sort these probabilities in a descending order and allocate them such that the largest probability goes with the smallest coefficient $c_i$ and so on. Assuming both the coefficients and the probabilities are known and sorted in advance, the complexity of this procedure is linear in $m$.

We now turn to the generalized BICA problem, defined in (\ref{eq:sum_ent_min_binary}). Since our objective is concave we would first like to bound it from above with a piecewise linear function which contains $k$ pieces, as shown in Figure \ref{fig:piecewise linear}. We show that solving the piecewise linear problem approximates the solution to (\ref{eq:sum_ent_min_binary}) as closely as we want.

\begin{figure}[!ht]
\centering
\includegraphics[width = 0.45\textwidth,bb= 85 272 500 520,clip]{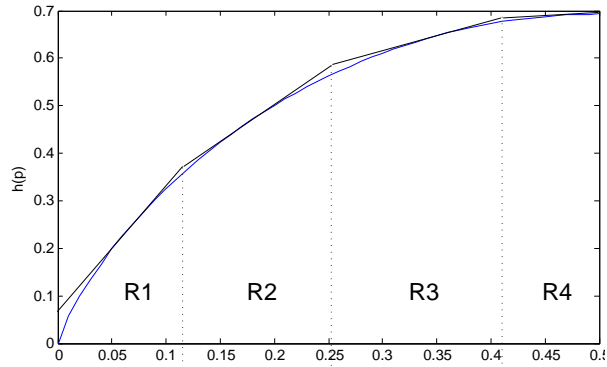}
\caption{piecewise linear ($k=4$) relaxation to the binary entropy}
\label{fig:piecewise linear}
\end{figure}

First, we notice that all $\pi_j's$ are exchangeable (in the sense that we can always interchange them and achieve the same result). This means we can find the optimal solution to the piece-wise linear problem by going over all possible combinations of ``placing" the $d$ variables $\pi_j$ in the $k$ different regions of the piece-wise linear function. For each of these combinations we need to solve a linear problem (\ref{eq:linear}), where the minimization is with respect to allocation of the  $m$ given probabilities $\underline{p}$, and with additional constraints on the ranges of each $\pi_j$. For example, assume $d=3$ and the optimal solution is such that two $\pi_j's$ (e.g. $\pi_1$ and $\pi_2$) are at the first region $R_1$ and $\pi_3$ is at the second region $R_2$ , then we need to solve the following constrained linear problem,

\begin{equation} \label{eq:linear_program}
\begin{aligned}
& {\text{minimize}}
& &  a_1 \cdot (\pi_1+\pi_2)+2b_1+a_2\cdot \pi_3 +b_2 \\
& \text{subject to}
& & \pi_1, \pi_2 \in R_1 , \pi_3 \in R_2
\end{aligned}
\end{equation}

where the minimization is over the allocation of the given $\left\{p_i\right\}_{i=1}^{m}$, which determine the corresponding $\pi_j$'s, as demonstrated in (\ref{eq:linear}). This problem again seems hard. However, if we attempt to solve it without the constraints we notice the following:

\begin{enumerate}

\item	If the collection of $\pi_j's$ which define the optimal solution to the unconstrained linear problem happens to meet the constraints then it is obviously the optimal solution with the constraints.
\item	If the collection of $\pi_j's$ of the optimal solution does not meet the constraints (say, $\pi_2 \in R_2$) then, due to the concavity of the entropy function, there exists a different combination with a different constrained linear problem,

\begin{equation*} \label{eq:linear_program_2}
\begin{aligned}
& {\text{minimize}}
& &  a_1\pi_1+b_1+a_2(\pi_2+\pi_3)+2b_2 \\
& \text{subject to}
& & \pi_1 \in R_1  \;  \pi_2,\pi_3 \in R_2
\end{aligned}
\end{equation*}
in which this set of $\pi_j's$ necessarily achieves a lower minimum (since $a_2 x+b_2<a_1 x+b_1$ $\forall x \in R_2$).
\end{enumerate}

Therefore, in order to find the optimal solution to the piece-wise linear problem, all we need to do is to go over all possible combinations of placing the $\pi_j's$ in $k$ different regions, and for each combination solve an unconstrained linear problem (which is solved in linear time in $m$). If the solution does not meet the constraint then it means that the assumption that the optimal $\pi_j$ reside within this combination's regions is false. Otherwise, if the solution does meet the constraint, it is considered as a candidate for the global optimal solution. 

The number of combinations we need to go through is equivalent to the number of ways of placing $d$ identical balls in $k$ boxes, which is (for a fixed $k$), 
\begin{equation}
\left(\begin{array}{c} 
d+k-1\\d\end{array}\right)=O(d^k).
\end{equation}
Assuming the coefficients are all known and sorted in advance, the overall complexity of our suggested algorithm, as $d \rightarrow \infty$, is just $O(d^k \cdot 2^d)=O(m\log^km)$.   

Notice that any approach which exploits the full statistical description of \underline{X}  would require going over the probabilities of all of its symbols at least once. Therefore, a computational load of at least $O(m)=O(2^d)$ seems inevitable. Still, this is significantly smaller  then $O(m!)=O(2^d!)$, required by brute-force search over all possible permutations.

It is also important to notice that even though the asymptotic complexity of our approximation algorithm is $O(m\log^km)$, it takes a few seconds to run an entire experiment on a standard personal computer (with $d=1024$ and $k=8$, for example). The reason is that the $m$ factor comes from the complexity of sorting a vector and multiplying two vectors, operations which are computationally efficient on most available software. Moreover, if we assume that the linear problems coefficients (\ref{eq:linear}) are calculated, sorted and stored in advance, we can place them in a matrix form and multiply the matrix with the (sorted) vector \underline{$p$}. The minimum of this product is exactly the solution to the linear approximation problem. Therefore, the practical asympthotic complexity of the approximation algorithm drops to a single multiplication of a ($\log^k(m) \times m$) matrix with a ($m \times 1$) vector.

Even though the complexity of this method is significantly lower than full enumeration, it may still be computationally infeasible as $m$ increases. Therefore, we suggest a simpler (greedy) solution, which is much easier to implement and apply.   

\subsection{Order Algorithm}
\label{ordering solution}
As mentioned above, the minimization problem we are dealing with (\ref{eq:sum_ent_min_binary}) is combinatorial in its essence and is consequently considered hard. We therefore suggest a simplified greedy algorithm which strives to sequentially minimize each term of the summation (\ref{eq:sum_ent_min_binary}), $H_b(Y_j)$, for $j=1,\dots,d$. 

With no loss of generality, let us start by minimizing $H_b(Y_1)$, which corresponds to the marginal entropy of the most significant bit (msb). Since the binary entropy is monotonically increasing in the range $\left[0,\frac{1}{2}\right]$, we would like to find a permutation of $\underline{p}$ that minimizes a sum of half of its values. This means we should order the $p_i$'s so that half of the $p_i$'s with the smallest values are assigned to $P(Y_1)=0$ while the other half of $p_i$'s (with the largest values) are assigned to $P(Y_1)=1$. For example, assuming $m=8$ and $p_1 \leq p_2 \leq \dots \leq p_8$, a permutation which minimizes $H_b(Y_1)$ is 
\begin{table}[!htbp]
\centering
{
\begin{tabular}{|c|c|c|c|c|c|c|c|c|}
\hline
codeword &000 &001 &010 &011&100 &101 &110 &111 \\
\hline
probability &$p_2$ &$p_3$&$p_1$&$p_4$&$p_8$ &$p_5$&$p_6$&$p_7$ \\
\hline
\end{tabular}}
\end{table}

We now proceed to minimize the marginal entropy of the second most significant bit, $H_b(Y_2)$. Again, we would like to assign $P(Y_2)=0$ the smallest possible values of $p_i$'s. However, since the we already determined which $p_i$'s are assigned to the msb, all we can do is reorder the $p_i$'s without changing the msb. This means we again sort the $p_i$'s so that the smallest possible values are assigned to  $P(Y_2)=0$, without changing the msb. In our example, this leads to,
\begin{table}[!htbp]
\centering
{
\begin{tabular}{|c|c|c|c|c|c|c|c|c|}
\hline
codeword &000 &001 &010 &011&100 &101 &110 &111 \\
\hline
probability &$p_2$ &$p_1$&$p_3$&$p_4$&$p_6$ &$p_5$&$p_8$&$p_7$ \\
\hline
\end{tabular}}
\end{table} 

Continuing in the same manner, we would now like to reorder the $p_i$'s to minimize $H_b(Y_3)$ without changing the previous bits. This results with

\begin{table}[!htbp]
\centering
{
\begin{tabular}{|c|c|c|c|c|c|c|c|c|}
\hline
codeword &000 &001 &010 &011&100 &101 &110 &111 \\
\hline
probability &$p_1$ &$p_2$&$p_3$&$p_4$&$p_5$ &$p_6$&$p_7$&$p_8$ \\
\hline
\end{tabular}}
\end{table} 

Therefore, we show that a greedy solution to (\ref{eq:sum_ent_min_binary}) which sequentially minimizes $H_b(Y_j)$ is attained by simply ordering the joint distribution $\underline{p}$ in an ascending (or equivalently descending) order. In other words, the  \textit{order permutation} suggests to  simply order the probability distribution $p_1, \dots, p_m$ in an ascending order, followed by a mapping of the $i^{th}$ symbol (in its binary representation) the $i^{th}$ smallest probability. 

At this point it seems quite unclear how well the order permutation performs, compared both with the relaxed BICA we previously discussed, and the optimal permutation which minimizes  (\ref{eq:sum_ent_min_binary}). In the following sections we introduce some theoretical properties which demonstrate its effectiveness.

\subsection{Worst-case Independent Representation}
\label{worst case}
We now introduce the theoretical properties of our suggested algorithms. Naturally, we would like to quantify how much we ``lose" by representing a given random vector $\underline{X}$ as if its components are statistically independent. Therefore, for any given random vector $\underline{X}  \sim \underline{p} $ and an invertible  transformation $\underline{Y}=g(\underline{X})$, we denote the cost function $C(\underline{p},g)=\sum_{j=1}^d H_b(Y_j) - H(\underline{X})$, as appears in (\ref{eq:sum_ent_min_binary}).

Since both our methods  strongly depend on the given probability distribution $\underline{p}$, we focus on the worst-case and the average case of $C(\underline{p},g)$, with respect to  $\underline{p}$. 
Let us denote the order permutation as $g_{ord}$ and the permutation which is found by the piece-wise linear relaxation as $g_{lin}$. We further define $g_{bst}$ as the permutation that results with a lower value of   $C(\underline{p},g)$, between $g_{lin}$ and $g_{ord}$. This means that $$g_{bst}=\underset{\{g_{lin},g_{ord}\}} {\arg\min}C(\underline{p},g).$$ 
In addition, we define $g_{opt}$ as the optimal permutation that minimizes (\ref{eq:sum_ent_min_binary}) over all possible permutations. Therefore, for any given $\underline{\tilde{p}}$, we have that $C(\underline{\tilde{p}},g_{opt}) \leq C(\underline{\tilde{p}},g_{bst}) \leq C(\underline{\tilde{p}},g_{ord})$. 
In this Section we examine the worst-case performance of both of our suggested algorithms. Specifically, we would like to quantify the maximum of $C(\underline{p},g)$ over all joint probability distributions $\underline{p}$, of a given alphabet size $m$. 

\begin{theorem}
\label{worst-case}
For any  random vector $\underline{X} \sim \underline{p}$, over an alphabet size $m$ we have that
$$\maxmax_{\underline{p}}C(\underline{p},g_{opt})=\Theta(\log(m))$$
\end{theorem}

\begin{proof}
We first notice that $\sum_{j=1}^d H_b(Y_j) \leq d=\log(m)$. In addition, $H(\underline{X}) \geq 0$. Therefore, we have that $C(\underline{p},g_{opt})$ is bounded from above by $\log(m)$. Let us also show that this bound is tight, in the sense that there  exists a joint probability distribution $\underline{\tilde{p}}$ such that  $C(\underline{\tilde{p}},g_{opt})$ is linear in $\log(m)$. 
Let $\tilde{p}_1=\tilde{p}_2=\dots=\tilde{p}_{m-1}=\frac{1}{3(m-1)}$ and $\tilde{p}_m=\frac{2}{3}$.
Then, $\underline{\tilde{p}}$ is ordered and satisfies $P(Y_i=0)=\frac{m}{6(m-1)}$.

In addition, we notice that assigning symbols in a decreasing order to $\underline{\tilde{p}}$ (as mentioned in Section \ref{ordering solution}) results with an optimal permutation. This is simply since  $P(Y_j=0)=\frac{m}{6(m-1)}$ is the minimal possible value of any $P(Y_j=0)$ that can be achieved when summing any $\frac{m}{2}$ elements of $\tilde{p}_i$.
Further we have that, 
\begin{align}
C(\underline{\tilde{p}},g_{opt})=&\sum_{j=1}^d H_b(Y_j) - H(\underline{X})=\\\nonumber
&\log(m)\cdot h_b\left( \frac{m}{6(m-1)}\right) +\left( (m-1) \frac{1}{3(m-1)}\log \frac{1}{3(m-1)} +\frac{2}{3}\log \frac{2}{3} \right)=\\\nonumber
&\log(m)\cdot h_b\left( \frac{m}{6(m-1)}\right)-\frac{1}{3}\log(m-1)+\frac{1}{3}\log\frac{1}{3}+\frac{2}{3}\log\frac{2}{3} \underset{m\rightarrow \infty}{\longrightarrow} \log(m) \cdot \left(h_b\left( \frac{1}{6}\right)-\frac{1}{3}\right)-h_b\left( \frac{1}{3}\right).
\end{align}
Therefore, $\maxmax_{\underline{p}}C(\underline{p},g_{opt})=\Theta(\log(m))$. \begin{flushright}
$\blacksquare$
\end{flushright}
\end{proof}

Theorem \ref{worst-case} shows that even the optimal permutation achieves a sum of marginal entropies which is $\Theta(\log(m))$ bits greater than the joint entropy, in the worst case. This means that there exists at least one source $\underline{X}$ with a joint probability distribution which is impossible to encode as if its components are independent without losing at least $\Theta(\log(m))$ bits. However, we now show that such sources are very ``rare".   

\subsection{Average-case Independent Representation}
\label{average case}
In this section we show that the expected value of  $C(\underline{p},g_{opt})$ is bounded by a small constant, when averaging uniformly over all possible $\underline{p}$ over an alphabet size $m$. 

To prove this, we recall that $C(\underline{p},g_{opt}) \leq C(\underline{p},g_{ord})$ for any given probability distribution $\underline{p}$. Therefore, we would like to find the expectation of $C(\underline{p},g_{ord})$ where the random variables are  $p_1, \dots, p_m$, taking values over a uniform simplex. 

\begin{proposition}
\label{H(X)}
Let $\underline{X} \sim \underline{p} $ be a random vector of an alphabet size $m$ and a joint probability distribution $\underline{p}$. The expected joint entropy of $\underline{X}$, where the expectation is  over a uniform simplex of joint probability distributions $\underline{p}$ is
\begin{equation}\nonumber
\mathbb{E}_{\underline{\smash{p}}}\left\{H(\underline{X}) \right\}=\frac{1}{\log_e{2}}\left(\psi(m+1)-\psi(2)\right) 
\end{equation}
where $\psi$ is the \textit{digamma function}.
\end{proposition}
The proof of this proposition is left for the Appendix. 

We now turn to examine the expected sum of the marginal entropies, $\sum_{j=1}^d H_b(Y_j)$ under the order permutation. As described above, the order permutation suggests sorting the probability distribution $p_1, \dots, p_m$ in an ascending order, followed by mapping of the $i^{th}$ symbol (in a binary representation) the $i^{th}$ smallest probability. Let us denote $p_{(1)}\leq \dots\leq p_{(m)}$ the ascending ordered probabilities $p_1, \dots, p_m$. Bairamov et al. \cite{bairamov2010limit} show that the expected value of $p_{(i)}$ is

\begin{equation}
\label{expectation}
\mathbb{E}\left\{p_{(i)}\right\}=\frac{1}{m}\sum_{k=m+1-i}^{m}\frac{1}{k}=\frac{1}{m}\left(K_m-K_{m-i}\right)
\end{equation}

where $K_m=\sum_{k=1}^{m}\frac{1}{k}$ is the Harmonic number. Denote the ascending ordered binary representation of all possible symbols in a matrix form $A \in \{0,1\}^{(m\times d)}$. This means that entry $A_{ij}$ corresponds to the $j^{th}$ bit in the $i^{th}$ symbol, when the symbols are given in an ascending order. Therefore, the expected sum of the marginal entropies of $\underline{Y}$, when the expectation is over a uniform simplex of joint probability distributions $p$, follows

\begin{equation}
\label{sum_of_marginals}
\mathbb{E}_{\underline{\smash{p}}}\left\{\sum_{j=1}^d H_b(Y_j) \right\} \underset{(a)}{\leq} \sum_{j=1}^d h_b(\mathbb{E}_{\underline{\smash{p}}}\{Y_j\}) \underset{(b)}= \sum_{j=1}^d  h_b \left(\frac{1}{m}\sum_{i=1}^{m}A_{ij}  \left(K_m-K_{m-i}\right)  \right) \underset{(c)}=\sum_{j=1}^d   h_b \left( \frac{1}{2}K_m- \frac{1}{m}\sum_{i=1}^{m}A_{ij}K_{m-i} \right)
\end{equation}
where $(a)$ follows from Jensen's inequality, $(b)$ follows from (\ref{expectation}) and $(c)$ follows since $\sum_{i=1}^{m}A_{ij}=\frac{1}{2}$ for all $j=1, \dots,d$.

We now turn to derive asymptotic bounds of the expected difference between the sum of $\underline{Y}$'s marginal entropies and the joint entropy of $\underline{X}$, as appears in (\ref{eq:sum_ent_min_binary}).

\begin{theorem}
\label{average case - asymp}
Let $\underline{X} \sim \underline{p}$ be a random vector of an alphabet size $m$ and joint probability distribution $\underline{p}$. Let $\underline{Y}=g_{ord}(\underline{X})$ be the order permutation. 
For $d \geq 10$, the expected value of $C(\underline{p},g_{ord})$, over a uniform simplex of joint probability distributions $\underline{p}$, satisfies
\begin{equation}
\nonumber
\mathbb{E}_{\underline{\smash{p}}}C(\underline{p},g_{ord})= \mathbb{E}_{\underline{\smash{p}}}\left\{\sum_{j=1}^d H_b(Y_j)- H(\underline{X}) \right\} < 0.0162 +O\left(\frac{1}{m}\right)
\end{equation}
\end{theorem}

\begin{proof}

Let us first derive the expected marginal entropy of the least significant bit, $j=1$, according to (\ref{sum_of_marginals}).

\begin{align}
\label{LSB}
\mathbb{E}_{\underline{\smash{p}}}\left\{ H_b(Y_1) \right\} \leq &  h_b \left( \frac{1}{2}K_m- \frac{1}{m}\sum_{i=1}^{m/2}K_{m-i} \right)=h_b \left( \frac{1}{2}K_m- \frac{1}{m}\left(\sum_{i=1}^{m-1}K_{i}-\sum_{i=1}^{\frac{m}{2}-1}K_{i} \right)\right)\underset{(a)}{=}\\\nonumber
&h_b \left( \frac{1}{2}K_m- \frac{1}{m}\left(mK_{m}-m-\frac{m}{2}K_{\frac{m}{2}}+\frac{m}{2} \right)\right)=h_b \left( \frac{1}{2}\left(K_{\frac{m}{2}}-K_{m}+1 \right)\right)\underset{(b)}{<}\\\nonumber
&h_b \left( \frac{1}{2}\log_{e}\left(\frac{1}{2}\right)+\frac{1}{2}+O\left(\frac{1}{m}\right)\right)\underset{(c)}{\leq} h_b \left(\frac{1}{2}\log_{e}\left(\frac{1}{2}\right)+\frac{1}{2}\right)+O\left(\frac{1}{m}\right) h_b' \left(\frac{1}{2}\log_{e}\left(\frac{1}{2}\right)+\frac{1}{2}\right)=\\\nonumber
&h_b \left(\frac{1}{2}\log_{e}\left(\frac{1}{2}\right)+\frac{1}{2}\right)+O\left(\frac{1}{m}\right)
\end{align}
 where $(a)$ and (b) follow the harmonic number properties:
\begin{enumerate} [(a)]

\item	$\sum_{i=1}^{m}K_{i}=(m+1)K_{m+1}-(m+1)$
\item	$\frac{1}{2(m+1)} < K_m-\log_e(m)-\gamma<\frac{1}{2m}$, where $\gamma$ is the Euler-Mascheroni constant \cite{young199175}
\end{enumerate}
and $(c)$ results from the concavity of the binary entropy.

Repeating the same derivation for different values of $j$, we attain

\begin{align}
\label{all_bits}
\mathbb{E}_{\underline{\smash{p}}}\left\{ H_b(Y_j) \right\} \leq &  h_b \left( \frac{1}{2}K_m-\frac{1}{m} \sum_{l=1}^{2^j-1} (-1)^{l+1}\sum_{i=1}^{l \frac{m}{2^j}}K_{m-i} \right)= h_b \left( \frac{1}{2}K_m-\frac{1}{m}\sum_{l=1}^{2^j} (-1)^{l}\sum_{i=1}^{l \frac{m}{2^j}-1}K_{i} \right)=\\\nonumber
& h_b \left( \frac{1}{2}K_m-\frac{1}{m}\sum_{l=1}^{2^j} (-1)^{l}\left(l  \frac{m}{2^j}K_{l  \frac{m}{2^j}}-l  \frac{m}{2^j}  \right) \right)<\\\nonumber
&  h_b \left( \sum_{i=1}^{2^j-1}(-1)^{i+1}\frac{i}{2^j}\log_{e}\left(\frac{i}{2^j}\right)+\frac{1}{2}\right)+O\left(\frac{1}{m} \right)\quad\quad \forall j=1,\dots,d.
\end{align}

We may now evaluate the sum of expected marginal entropies of $\underline{Y}$.
For simplicity of derivation let us obtain $\mathbb{E}_{\underline{\smash{p}}}\left\{ H_b(Y_j) \right\}$ for $j=1, \dots, 10$ according to (\ref{all_bits}) and upper bound $\mathbb{E}_{\underline{\smash{p}}}\left\{ H_b(Y_j) \right\}$ for $j>10$ with $h_b\left(\frac{1}{2}\right)=1$.  This means that for $d \geq 10$ we have 

\begin{equation}
 \mathbb{E}_{\underline{\smash{p}}}\left\{\sum_{j=1}^d H_b(Y_j)\right\} <
\sum_{j=1}^{10} \mathbb{E}_{\underline{\smash{p}}}\left(H_b\left\{Y_j\right\}\right) +\sum_{j=11}^d h_b\left(\frac{1}{2}\right)< 9.4063 + (d-10)+O\left(\frac{1}{m} \right).
\end{equation}

The expected joint entropy may also be expressed in a more compact manner. In Proposition \ref{H(X)} it is shown than $\mathbb{E}_{\underline{\smash{p}}}\left\{H(\underline{X}) \right\}=\frac{1}{\log_e{2}}\left(\psi(m+1)-\psi(2)\right) $. Following the inequality in \cite{young199175}, the Digamma function, $\psi(m+1)$, is bounded from below by $\psi(m+1)=H_m-\gamma>\log_e(m)+\frac{1}{2(m+1)}$. Therefore, we conclude that for $d\geq10$ we have that  

\begin{equation}
 \mathbb{E}_{\underline{\smash{p}}}\left\{\sum_{j=1}^d H_b(Y_j)-H(\underline{X})\right\} < 9.4063 + (d-10)-\log{(m)} + \frac{\psi(2)}{\log_e{2}} +O\left(\frac{1}{m}\right) = 0.0162 +O\left(\frac{1}{m}\right)
\end{equation}
\begin{flushright}
$\blacksquare$
\end{flushright}
\end{proof}

In addition, we would like to evaluate the expected difference between the sum of marginal entropies and the joint entropy of $\underline{X}$, that is, without applying any permutation. This shall serve us as a reference  to the upper bound we achieve in Theorem \ref{average case - asymp}. 

\begin{theorem}
\label{theorem4}
Let $\underline{X} \sim \underline{p}$ be a random vector of an alphabet size $m$ and joint probability distribution $\underline{p}$. 
The expected difference between the sum of marginal entropies and the joint entropy of $\underline{X}$, when the expectation is taken over a uniform simplex of joint probability distributions $\underline{p}$, satisfies
\begin{equation}
\nonumber
 \mathbb{E}_{\underline{\smash{p}}}\left\{\sum_{j=1}^d H_b(X_j)- H(\underline{X}) \right\} < \frac{\psi(2)}{\log_e{2}}=0.6099
\end{equation}
\end{theorem}

\begin{proof}
We first notice that $P\left(X_j=1\right)$ equals the sum of one half of the probabilities $p_i, i=1, \dots, m$ for every $j=1 \dots d$. Assume $p_i$'s are randomly (and uniformly) assigned to each of the $m$ symbols. Then, $\mathbb{E}\{P\left(X_j=1\right)\}=\frac{1}{2}$ for every $j=1 \dots d$. Hence,

\begin{align}
\nonumber
 \mathbb{E}_{\underline{\smash{p}}}\left\{\sum_{j=1}^d H_b(X_j)- H(\underline{X}) \right\} =& 
\sum_{j=1}^d \mathbb{E}_{\underline{\smash{p}}}\left\{H_b(X_j)\right\}- \mathbb{E}_{\underline{\smash{p}}}\{H(\underline{X})\} <
d-\log{(m)}+\frac{1}{\log_e{2}} \left( \psi(2)-\frac{1}{2(m+1)}\right)<\frac{\psi(2)}{\log_e{2}}
\end{align}
\begin{flushright}
$\blacksquare$
\end{flushright}
\end{proof}

To conclude, we show that for a random vector $\underline{X}$ over an alphabet size $m$, we have 

$$ \mathbb{E}_{\underline{\smash{p}}}C(\underline{p},g_{opt}) \leq \mathbb{E}_{\underline{\smash{p}}}C(\underline{p},g_{bst}) \leq \mathbb{E}_{\underline{\smash{p}}}C(\underline{p},g_{ord}) < 0.0162+O\left(\frac{1}{m}\right)$$

for $d\geq10$, where the expectation is over a uniform simplex of joint probability distributions $\underline{p}$. 

This means that when the alphabet size is large enough, even the simple order permutation achieves, on the average, a sum of marginal entropies which is only $0.0162$ bits greater than the joint entropy, when all possible probability distributions $\underline{p}$ are equally likely to appear. Moreover, we show that the simple order permutation reduced the expected difference between the sum of the marginal entropies and the joint entropy of $\underline{X}$ by more than half a bit, for sufficiently large $m$.

\section{Large Alphabet Source Coding}
\label{classic source coding}
Assume a classic compression setup in which both the encoder and the decoder are familiar with the joint probability distribution of the source $\underline{X} \sim \underline{p}$, and the number of observations $n$ is sufficiently large in the sense that $\hat{H}(\underline{X})\approx H(\underline{X})$. 

As discussed above, both Huffman and arithmetic coding entail a growing redundancy and a quite involved implementation as the alphabet size increases. 
The Huffman code guarantees a redundancy of at most a single bit for every alphabet size, depending on the dyadic structure of $p$. On the other hand, arithmetic coding does not require a dyadic $p$, but only guarantees a redundancy of up to two bits, and is practically limited for smaller alphabet size \cite{cover2012elements,yang2000universal}.

In other words, both Huffman and arithmetic coding are quite likely to have an average codeword length which is greater than $H(\underline{X})$, and are complicated (or sometimes even impossible) to implement, as $m$ increases.

To overcome these drawbacks, we suggest a simple solution in which we first apply an invertible transformation to make the components of $\underline{X}$ ``as statistically independent as possible", following an entropy coding on each of its components separately. This scheme results with a redundancy which we previously defined as $C(\underline{p},g)=\sum_{j=1}^m H(Y_j) -H(\underline{X})$. However, it allows us to apply a Huffman or arithmetic encoding on each of the components separately; hence, over a binary alphabet. 

Moreover, notice we can group several components, $Y_j$, into blocks so that the joint entropy of the block is necessarily lower than the sum of marginal entropies of $Y_j$. Specifically, denote $b$ as the number of components in each block and $B$ as the number of blocks. Then, $b \times B =d$ and for each block $v =1, \dots, B$ we have that
\begin{equation}
\label{block inequality}
H(\underline{Y}^{(v)}) \leq \sum_{u=1}^b H_b(Y_u^{(v)})
\end{equation}
where $H(\underline{Y}^{(v)})$ is the entropy of the block $v$ and $H_b(Y_u^{(v)})$ is the marginal entropy of the $u^{th}$ component of the block $v$. Summing over all $B$ blocks we have 
\begin{equation}
\label{total inequality}
\sum_{v=1}^B H(\underline{Y}^{(v)}) \leq \sum_{v=1}^B\sum_{u=1}^b H_b(Y_u^{(v)})=\sum_{j=1}^d H_b(Y_j).
\end{equation}
This means we can always apply our suggested invertible transformation which minimizes 
$\sum_{j=1}^d H_b(Y_j)$,  and then the group components into $B$ blocks and encode each block separately. This results with  $\sum_{v=1}^B H(\underline{Y}^{(v)}) \leq \sum_{j=1}^d H_b(Y_j)$. By doing so, we increase the alphabet size of each block (to a point which is still not problematic to implement with Huffman or arithmetic coding) while at the same time we decrease the redundancy. We discuss different considerations in choosing the number of blocks $B$ in the following sections.

A more direct approach of minimizing the sum of block entropies $\sum_{v=1}^B H(\underline{Y}^{(v)})$ is to refer to each block as a symbol over a greater alphabet size, $2^b$. This allows us to seek an invertible transformation which minimizes the sum of marginal entropies, where each marginal entropy corresponds to a marginal probability distribution over an alphabet size $2^b$. This minimization problem is referred to as a generalized ICA over finite alphabets and is discussed in detail in \cite{painsky2015generalized}.

However,  notice that both the Piece-wise Linear Relaxation algorithm (Section \ref{piece-wise solution}), and the solutions discussed in \cite{painsky2015generalized}, require an extensive computational effort in finding a minimizer for (\ref{eq:sum_ent_min_binary}) as the alphabet size increases. Therefore, we suggest applying the greedy order permutation as $m$ grows. This solution may result in quite a large redundancy for a several joint probability distributions $\underline{p}$ (as shown in Section \ref{worst case}). However, as we uniformly average over all possible $p$'s, the redundancy is bounded with a small constant as the alphabet size increases (Section  \ref{average case}).

Moreover, the ordering approach simply requires ordering the values of  $\underline{p}$, which is significantly faster than constructing a Huffman dictionary or arithmetic encoder.

To illustrate our suggested scheme, consider a source $\underline{X} \sim \underline{p}$ over an alphabet size $m$, which follows the Zipf's law distribution, 
\begin{equation}\nonumber
P(k;s,m)=\frac{k^{-s}}{\sum_{l=1}^m l^{-s}}
\end{equation}
where $m$ is the alphabet size and $s$ is the ``skewness" parameter. The Zipf's law distribution is a commonly used heavy-tailed distribution, mostly in modeling of natural (real-world) quantities. It is widely used in physical and social sciences, linguistics, economics and many other fields.   

We would like to design an entropy code for $\underline{X}$ with $m=2^{16}$ and different values of $s$.
We first apply a standard Huffman code as an example of a common entropy coding scheme. Notice that we are not able to construct an arithmetic encoder as the alphabet size is too large \cite{yang2000universal}. We further apply our suggested order permutation scheme (Section \ref{ordering solution}), in which we sort $\underline{p}$ in a descending order, followed by arithmetic encoding to each of the components separately.  We further group these components into two separate blocks (as discussed above) and apply an arithmetic encoder on each of the blocks.     
We repeat this experiment for a range of parameter values $s$. Figure \ref{fig:classic_compression} demonstrates the results we achieve.

\begin{figure}[ht]
\centering
\includegraphics[width = 0.7\textwidth,bb= 60 235 540 545,clip]{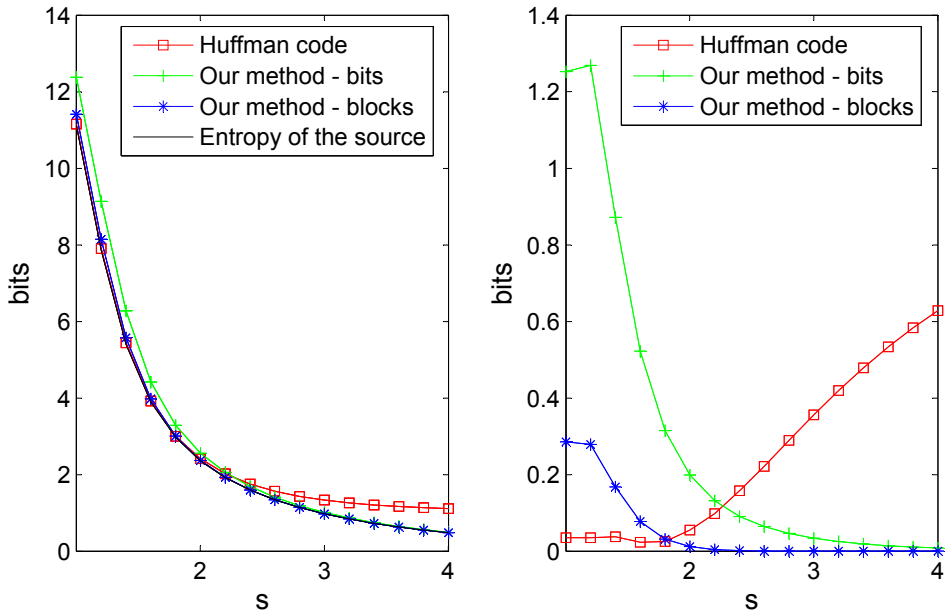}
\caption{Zipf's law simulation results. Left: the curve with the squares is the average codeword length using a Huffman code, the curve with the crosses corresponds to the average codeword length using our suggested methods when encoding each component separately, and the curve with the asterisks is our suggested method when encoding each of the two blocks separately. The black curve (which tightly lower-bounds all the curves) is the entropy of the source. Right: The difference between each encoding method and the entropy of the source}
\label{fig:classic_compression}
\end{figure}

Our results show that the Huffman code attains an average codeword length which is very close to the entropy of the source for lower values of $s$. However, as $s$ increases and the distribution of the source becomes more skewed, the Huffman code diverges from the entropy. On the other hand, our suggested method succeeds in attaining an average codeword length which is very close to the entropy of the source for every $s$, especially as $s$ increases and when independently encoding each of the blocks.

\section{Universal Source Coding}
\label{universal source coding}
The classical source coding problem is typically concerned with a source whose alphabet size is much smaller than the length of the sequence. In this case one usually assumes that $\hat{H}(\underline{X}) \approx H(\underline{X})$.   
However, in many real world applications such an assumption is not valid. A paradigmatic example is the word-wise compression of natural language texts.
In this setup we draw a memoryless sequence of words, so that the alphabet size is often comparable to or even larger than the length of the source sequence. 

As discussed above, the main challenge in large alphabet source coding is the redundancy of the code, which is formally defined as the excess number of bits used over the source's entropy. The redundancy may be quantified as the expected number of extra bits required to code a memoryless sequence drawn from $\underline{X} \sim \underline{p}$, when using a code that was constructed for $\underline{p}$, rather than using the ``true" code, optimized for the empirical distribution $\hat{\underline{p}}$.  
Another way to quantify these extra bits is to directly design a code for $\hat{\underline{p}}$, and transmit the encoded sequence together with this code.

Here again, we claim that in some cases, applying a transformation which decomposes the observed sequence into multiple ``as independent as possible" components results in a better compression rate.

However, notice that now we also need to consider the number of bits required to describe the transformation. In other words, our redundancy involves not only (\ref{eq:min_criterion}) and the designated code for the observed sequence, but also the invertible transformation we applied on the sequence. This means that even the simple order permutation (Section \ref{ordering solution}) requires at most $n\log{m}$ bits to describe, where $m$ is the alphabet size and $n$ is the length of the sequence. This redundancy alone is not competitive with  Szpankowski and Weinberger \cite{szpankowski2012minimax} worst-case redundancy results, described in (\ref{n=o(m)}).

Therefore, we require a different approach which minimizes the sum of marginal entropies (\ref{eq:sum_ent_min_binary}) but at the same time is simpler to describe.

One possible solution is to seek invertible, yet linear, transformations. This means that describing the transformation would now only require $\log^2{m}$ bits. However, this generalized linear BICA problem is also quite involved. In their work, Attux et al. \cite{attux2011immune} describe the difficulties in minimizing (\ref{eq:sum_ent_min_binary}) over XOR field (linear transformations) and suggest an immune-inspired algorithm for it. Their algorithm, which is heuristic in its essence, demonstrates some promising results. However, it is not very scalable (with respect to the number of components $\log{m}$) and does not guarantee to converge to the global optimal solution.

Therefore, we would like to modify our suggested approach (Section \ref{piece-wise solution}) so that the transformation we achieve requires fewer bits to describe.


As in the previous section, we argue that in some setups it is better to split the components of the data into blocks, with $b$ components in each block, and encode the blocks separately. Notice that we may set the value of $b$ so that the blocks are no longer considered as over a large alphabet size $(n \gg 2^b)$. This way, the redundancy of encoding each block separately is again negligible, at the cost of longer averaged codeword length. For simplicity of notation we define the number of blocks as $B$, and assume $B=\sfrac{d}{b}$ is a natural number. Therefore, encoding the $d$ components all together takes  $n\cdot \hat{H}(\underline{X})$ bits for the data itself, plus a redundancy term according to (\ref{m=o(n)}) and (\ref{n=o(m)}), while the block-wise compression takes about
\begin{equation}
\label{compression_blocks}
n\cdot \sum_{v=1}^{B}{ \hat{H}({\underline{X}}^{(v)})}+B\frac{2^b-1}{2}\log{\frac{n}{2^b}}
\end{equation}
bits, where the first term is $n$ times the sum of $B$ empirical block entropies and the second term is $B$ times the redundancy of each block when $n=o(2^b)$.
Two subsequent questions arise from this setup:
\begin{enumerate}

\item	What is the optimal value of $b$ that minimizes (\ref{compression_blocks})?
\item	Given a fixed value of $b$, how can we re-arrange $d$ components into $B$ blocks so that the averaged codeword length (which is bounded from below by the empirical entropy), together with the redundancy, is as small as possible?

\end{enumerate}
Let us start by fixing $b$ and focusing on the second question. 

A naive shuffling approach is to exhaustively or randomly search for all possible combinations of clustering $d$ components into $B$ blocks. Assuming $d$ is quite large, an exhaustive search is practically infeasible. Moreover, the shuffling search space is quite limited and results with a very large value of (\ref{eq:min_criterion}), as shown below. Therefore, a different method is required.
We suggest applying our generalized BICA tool as an upper-bound search method for efficiently searching for a minimal possible averaged codeword length.
As in previous sections we define $\underline{Y}=g(\underline{X})$, where $g$ is some invertible transformation of $\underline{X}$.
Every block of the vector $\underline{Y}$ satisfies (\ref{block inequality}), where the entropy terms are now replaced with empirical entropies.  In the same manner as in Section \ref{classic source coding}, summing over all $B$ blocks results with  (\ref{total inequality}) where again, the entropy terms are replaced with empirical entropies.
This means that the sum of the empirical block entropies is  bounded from above by the empirical marginal entropies of the components of $\underline{Y}$ (with equality iff the components are independently distributed).

\begin{equation}
\label{universal total inequality}
\sum_{v=1}^B \hat{H}(\underline{Y}^{(v)}) \leq \sum_{j=1}^d \hat{H}_b(Y_j).
\end{equation}

Our suggested scheme works as follows: 
We first randomly partition the $d$ components into $B$ blocks. We estimate the joint probability of each block and apply the generalized BICA on it. The sum of empirical marginal entropies (of each block) is an upper bound on the empirical entropy of each block, as described in the previous paragraph. Now, let us randomly shuffle the $d$ components of the vector $\underline{Y}$. By ``shuffle" we refer to an exchange of positions of $\underline{Y}$'s components. Notice that by doing so, the sum of empirical marginal entropies of the entire vector $\sum_{i=1}^{d}{\hat{H}_b(Y_i)}$ is maintained. We now apply the generalized BICA on each of the (new) blocks. This way we minimize (or at least do not increase) the sum of empirical marginal entropies of the (new) blocks.  This obviously results with a lower sum of empirical marginal entropies of the entire vector $\underline{Y}$. It also means that we minimize the left hand side of (\ref{universal total inequality}), which upper bounds the sum of empirical block entropies, as the inequality in (\ref{universal total inequality}) suggests. In other words, we show that in each iteration we decrease (at least do not increase) an upper bound on our objective. We terminate once a maximal number of iterations is reached or we can no longer decrease the sum of empirical marginal entropies.

Therefore, assuming we terminate at iteration $I_0$, encoding the data takes about 
\begin{align}
\label{compression_total}
   n\cdot \sum_{v=1}^{B}{\hat{H}^{[I_0]}({\underline{Y}}^{(v)})}+B\frac{2^b-1}{2}\log{\frac{n}{2^b}}+ I_0B\cdot &b2^b+I_0d\log{d}
\end{align}
bits, where the first term refers to the sum of empirical block entropies at the $I_0$ iteration, the third term refers to the representation of $I_0 \cdot B$ invertible transformation of each block during the process until $I_0$, and the fourth term refers to the bit permutations at the beginning of each iteration.

Hence, to minimize (\ref{compression_total}) we need to find the optimal trade-off between a low value of $\sum_{v=1}^{B}{\hat{H}^{[I_0]}({\underline{Y}}^{(v})}$ and a low iteration number $I_0$. 
We may apply this technique with different values of $b$ to find the best compression scheme over all block sizes.

\subsection{synthetic experiments}

In order to demonstrate our suggested method we first generate a dataset according to the Zipf law distribution which was previously described.  We draw $n=10^6$ realizations from this distribution with an alphabet size $m=2^{20}$ and a parameter value $s=1.2$. We encounter $n_0=80,071$ unique words and attain an empirical entropy of $8.38$ bits (while the true entropy is $8.65$ bits). Therefore, compressing the drawn realizations in its given $2^{20}$ alphabet size takes a total of about $10^{6}\times 8.38+1.22\times10^6=9.6 \cdot 10^6$ bits, according to (\ref{m=theta(n)}). 
Using the patterns method \cite{orlitsky2004universal}, the redundancy we achieve is the redundancy of the pattern plus the size of the dictionary. Hence, the compressed size of the data set according to this method is lower bounded by $10^{6}\times8.38+80,071\times20+100=9.982\cdot10^6$ bits.
In addition to these asymptotic schemes we would also like to compare our method with a common practical approach. For this purpose we apply the canonical version of the Huffman code. Through the canonical Huffman code we are able to achieve a compression rate of $9.17$ bits per symbol, leading to a total compression size of about $1.21\cdot10^7$ bits.   

Let us now apply a block-wise compression. 
We first demonstrate the behavior of our suggest approach with four blocks $(B=4)$ as appears in Figure \ref{fig:zipf}.
To have a good starting point, we initiate our algorithm with a the naive shuffling search method (described above). This way we apply our optimization process on the best representation a random bit shuffling could attain (with a negligible $d\log{d}$ redundancy cost).  As we can see in Figure \ref{fig:zipf}.B, we  minimize (\ref{compression_total}) over $I_0=64$ and $\sum_{v=1}^{B}{\hat{H}({\underline{Y}}^{(v)})}=9.09$  to achieve a total of $9.144 \cdot 10^6$ bits for the entire dataset. 

Table \ref{table:zipf_results} summarizes the results we achieve for different block sizes $B$. We see that the lowest compression size is achieved over $B=2$, i.e. two blocks. The reason is that for a fixed $n$, the redundancy is approximately exponential in the size of the block $b$. This means the redundancy drops exponentially with the number of blocks while the minimum of $\sum_{v=1}^{B}{\hat{H}({\underline{Y}}^{(v)})}$ keeps increasing. In other words, in this example we earn a great redundancy reduction when moving to a two-block representation while not losing too much in terms of the average code-word length we can achieve. We further notice that the optimal iterations number grows with the number of blocks. This results from the cost of describing the optimal transformation for each block, at each iteration, $I_0B\cdot b2^b$, which exponentially increase with the block size $b$. Comparing our results with the three methods described above we are able to reduce the total compression size in $8\cdot10^5$ bits, compared to the minimum among all our competitors. 
  
\begin{figure}[!ht]
\centering
\includegraphics[width = 0.6\textwidth,bb= 35 123 775 490,clip]{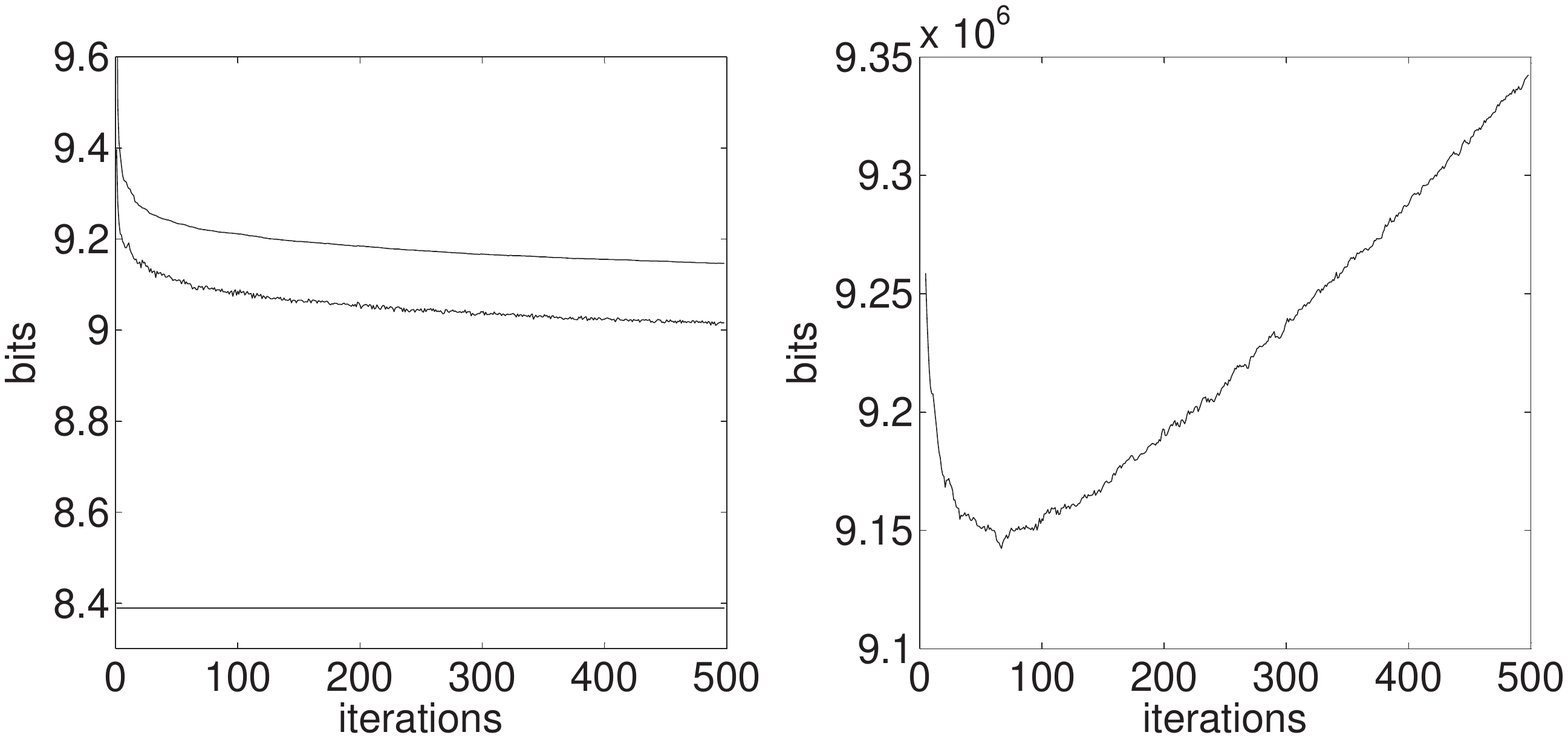}
\caption{Large Alphabet Source Coding via Generalized BICA with $B=4$ blocks. Left side (A): the horizontal line indicated the empirical entropy of $\underline{X}$. The upper curve is the sum of marginal empirical entropies and the lower curve is the sum of empirical block entropies (the outcome of our suggested framework). Right side (B): total compression size of our suggested method at each iteration.}
\label{fig:zipf}
\end{figure}

\begin{table}[!ht]
\caption{Block-Wise Compression via Generalized BICA Method for different block sizes}
\renewcommand{\baselinestretch}{1}\footnotesize
\label{table:zipf_results} 
\centering
\begin{tabular}{|M{1.8cm}|M{2.5cm}|M{1.8cm}|M{2cm}|M{1.8cm}|M{2.8cm}|N}  

\hline

\begin{tabular}{@{}c@{}}Number of \\ Blocks\end{tabular}
&\begin{tabular}{@{}c@{}} Minimum of \\ $\sum_{v=1}^{B}{\hat{H}({\underline{Y}}^{(v)})}$\end{tabular}
&Optimal $I_0$
&\begin{tabular}{@{}c@{}} Compressed \\  Data Size\end{tabular}
&Redundancy
&\begin{tabular}{@{}c@{}} Total Compression\\ Size\end{tabular}&\\[15pt]

\hline
$2$   &$8.69$ & $5$ & $8.69\cdot 10^6$ & $1.15\cdot 10^5$  & $\bold{8.805\cdot10^6}$ & \\[10pt]  
\hline
$3$   & $8.93$ & $19$ & $8.93\cdot 10^6$ & $5.55\cdot 10^4$  & $8.985\cdot10^6$ & \\[10pt]  
\hline
$4$   & $9.09$ & $64$ & $9.09\cdot 10^6$ & $5.41\cdot 10^4$  & $9.144\cdot10^6$ & \\[10pt] 
\hline
\end{tabular}
\end{table}

\subsection{real-world experiments}

We now turn to demonstrate our compression framework on real world data sets.
For this purpose we use collections of word frequencies of different natural languages. These word frequency lists are publicly available\footnote{\url{http://en.wiktionary.org/wiki/Wiktionary:Frequency_lists}} and describe the frequency each word appears in a language, based on hundreds of millions of words, collected from open source subtitles\footnote{\url{www.opensubtitles.org}} or based on different dictionaries and glossaries \cite{new2004lexique}.
Since each word frequency list holds several hundreds of thousands of different words, we choose a binary $d=20$ bit representation. We sample $10^7$ words from each language and examine our suggested framework, compared with the compression schemes mentioned above. The results we achieve are summarized in Table \ref{table:languages}. Notice the last column provides the percentage of the redundancy we save, which is essentially the most we can hope for (as we cannot go lower than $n\cdot \hat{H}(\underline{X})$ bits). As in the previous experiment, our suggested algorithm achieves the lowest compression size applied with two blocks after approximately $I_0=10$ iterations, from the same reasons mentioned above. 
Compared to the other methods, our suggested framework shows to achieve significantly lower compression sizes for all languages, saving an average of over one million bits per language.

\begin{table}[ht]
\caption{Natural Languages Experiment. For each compression method (D), (O) and (T) stand for the compressed data, the overhead and the total compression size (in bits) respectively. The We Save column is the amount of bits saved by our method, and its corresponding percentage of (O) and (T). $n_0$ is the number of unique words observed in each language, of the $10^7$ sampled words. Notice the Chinese corpus refers to characters.\newline}

\renewcommand{\baselinestretch}{1}\footnotesize
\label{table:languages} 
\centering
\begin{tabular}{|M{1.42cm}|M{2.3cm}|M{2.3cm}|M{2.3cm}|M{2.3cm}|M{1.95cm}|N}  

\hline

\begin{tabular}{@{}c@{}} Language \\ $(n_0)$ \end{tabular}
&\begin{tabular}{@{}c@{}} Standard \\ Compression \end{tabular}
&\begin{tabular}{@{}c@{}} Patterns \\ Compression \end{tabular}
&\begin{tabular}{@{}c@{}} Canonical \\  Huffman \end{tabular}
&\begin{tabular}{@{}c@{}} Our Suggested \\  Method \end{tabular}
&We Save&\\[20pt]

\hline

\begin{tabular}{@{}c@{}} English \\ $(129,834)$ \end{tabular}   
&\begin{tabular}{@{}c@{}} (D) $9.709\cdot10^{7}$ \\ (O) $2.624 \cdot10^{6}$\\ (T)  $9.971\cdot10^{7}$ \end{tabular} 
&\begin{tabular}{@{}c@{}} (D) $9.709\cdot10^{7}$ \\ (O) $2.597 \cdot10^{6}$ \\ (T)  $9.968\cdot10^{7}$ \end{tabular} 
&\begin{tabular}{@{}c@{}} (D) $9.737\cdot10^{7}$ \\ (O)  $5.294 \cdot10^{6}$ \\ (T)  $1.027 \cdot10^{8}$ \end{tabular} 
&\begin{tabular}{@{}c@{}} (D) $9.820 \cdot10^{7}$ \\ (O) $2.207 \cdot10^{5}$ \\ (T) $\bold{9.842 \cdot10^{7}}$ \end{tabular}  
&\begin{tabular}{@{}c@{}} $1.262 \cdot10^{6}$ \\  (O) $48.6\%$  \\ (T) $1.27\%$   \end{tabular}    & \\[32pt]

\hline

\begin{tabular}{@{}c@{}} Chinese \\ $(87,777)$ \end{tabular}     
&\begin{tabular}{@{}c@{}} (D)  $1.020\cdot10^{8}$\\ (O)  $2.624\cdot10^{6}$ \\ (T)   $1.046\cdot10^{8}$ \end{tabular} 
&\begin{tabular}{@{}c@{}} (D)  $1.020\cdot10^{8}$ \\ (O) $1.696\cdot10^{6}$ \\ (T)  $1.037\cdot10^{8}$ \end{tabular} 
&\begin{tabular}{@{}c@{}} (D) $1.023\cdot10^{8}$ \\ (O) $3.428\cdot10^{6}$ \\ (T) $1.057\cdot10^{8}$ \end{tabular} 
&\begin{tabular}{@{}c@{}} (D) $1.028\cdot10^{8}$ \\ (O) $2.001\cdot10^{5}$ \\ (T)  $\bold{1.030\cdot10^{8}}$ \end{tabular}   
&\begin{tabular}{@{}c@{}} $6.566 \cdot10^{5}$ \\  (O) $38.7\%$ \\ (T) $0.63\%$\end{tabular}    & \\[32pt]   
\hline

\begin{tabular}{@{}c@{}} Spanish \\ $(185,866)$ \end{tabular}     
&\begin{tabular}{@{}c@{}} (D)  $1.053\cdot10^{8}$\\ (O)  $2.624\cdot10^{6}$ \\ (T)   $1.079\cdot10^{8}$ \end{tabular} 
&\begin{tabular}{@{}c@{}} (D)  $1.053\cdot10^{8}$ \\ (O) $3.718\cdot10^{6}$ \\ (T)  $1.090\cdot10^{8}$ \end{tabular} 
&\begin{tabular}{@{}c@{}} (D)  $1.055\cdot10^{8}$ \\ (O) $7.700\cdot10^{6}$ \\ (T)  $1.132\cdot10^{8}$  \end{tabular} 
&\begin{tabular}{@{}c@{}} (D) $1.067\cdot10^{8}$ \\ (O) $2.207\cdot10^{5}$ \\ (T)  $\bold{1.069\cdot10^{8}}$ \end{tabular}   
&\begin{tabular}{@{}c@{}} $9.631 \cdot10^{5}$ \\  (O) $36.7\%$ \\ (T) $0.89\%$\end{tabular}    & \\[32pt]   
\hline

\begin{tabular}{@{}c@{}} French \\ $(139,674)$ \end{tabular}     
&\begin{tabular}{@{}c@{}} (D)  $1.009\cdot10^{8}$\\ (O)  $2.624\cdot10^{6}$ \\ (T)   $1.035\cdot10^{8}$ \end{tabular} 
&\begin{tabular}{@{}c@{}} (D)  $1.009\cdot10^{8}$ \\ (O) $2.794\cdot10^{6}$ \\ (T)  $1.036\cdot10^{8}$ \end{tabular} 
&\begin{tabular}{@{}c@{}} (D)  $1.011\cdot10^{8}$ \\ (O)  $5.745\cdot10^{6}$  \\ (T)  $1.069\cdot10^{8}$ \end{tabular} 
&\begin{tabular}{@{}c@{}} (D) $1.017\cdot10^{8}$ \\ (O) $2.207\cdot10^{5}$ \\ (T)  $\bold{1.019\cdot10^{8}}$ \end{tabular}   
&\begin{tabular}{@{}c@{}} $1.557 \cdot10^{6}$ \\  (O) $59.3\%$\\ (T) $1.50\%$ \end{tabular}    & \\[32pt]   
\hline

\begin{tabular}{@{}c@{}} Hebrew \\ $(250,917)$ \end{tabular}     
&\begin{tabular}{@{}c@{}} (D)  $1.173\cdot10^{8}$\\ (O)  $2.624\cdot10^{6}$ \\ (T)   $1.200\cdot10^{8}$ \end{tabular} 
&\begin{tabular}{@{}c@{}} (D)  $1.173\cdot10^{8}$ \\ (O) $5.019\cdot10^{6}$ \\ (T)  $1.224\cdot10^{8}$ \end{tabular} 
&\begin{tabular}{@{}c@{}} (D)  $1.176\cdot10^{8}$ \\ (O)  $1.054\cdot10^{7}$  \\ (T)  $1.281\cdot10^{8}$ \end{tabular} 
&\begin{tabular}{@{}c@{}} (D) $1.190\cdot10^{8}$ \\ (O) $1.796\cdot10^{5}$ \\ (T)  $\bold{1.192\cdot10^{8}}$ \end{tabular}   
&\begin{tabular}{@{}c@{}} $7.837 \cdot10^{5}$ \\  (O) $29.9\%$ \\ (T) $0.65\%$  \end{tabular}    & \\[32pt]   
\hline

\end{tabular}
\end{table}
\vspace{1em}

\section{Vector Quantization}

Vector quantization refers to a lossy compression setup, in which a high-dimensional vector $\underline{X} \in \mathbb{R}^d$ is to be represented by a finite number of points. This means that the high dimensional observed samples are clustering into groups, where each group is represented by a representative point. For example, the famous $k$-means algorithm \cite{macqueen1967some} provides a method to determine the clusters and the representative points (centroids) for an Euclidean loss function. Then, these centroid points that represent the observed samples are compressed in a lossless manner. 

As described above, in the lossy encoding setup one is usually interested in minimizing the amount of bits which represent the data for a given a distortion (or equivalently, minimizing the distortion for a given compressed data size).  The rate-distortion function defines the lower bound on this objective. In vector quantization, the representation is a deterministic mapping (defined as $P(\underline{Y}|\underline{X})$) from a source $\underline{X}$ to its quantized version $\underline{Y}$. Therefore we have that $H(\underline{Y}|\underline{X})=0$ and the rate distortion is simply   

\begin{equation}
\label{R(D)}
R\left(D\right)=\min_{P(\underline{Y}|\underline{X})}H(\underline{Y})\,\, s.t. \,\, \mathbb{E}\left\{D(\underline{X},\underline{Y})\right\} \leq D
\end{equation}

where $D(\underline{X},\underline{Y})$ is some distortion measure between $\underline{X}$ and $\underline{Y}$. 

\subsection{Entropy Constrained Vector Quantization}

The Entropy Constrained Vector Quantization (ECVQ) is an iterative method for clustering the observed samples into centroid points which are later represented by a minimal average codeword length. The ECVQ algorithm aims to find the minimizer of  

\begin{equation}
\label{ECVQ minimization}
J\left(D\right)=\min \mathbb{E}\left\{l(\underline{X})\right\} \,\, s.t. \,\, \mathbb{E}\left\{D(\underline{X},\underline{Y})\right\} \leq D
\end{equation}
where the minimization is over three terms: the vector quantizer (of $\underline{X}$), the entropy encoder (of the quantized version of $\underline{X}$) and the reconstruction module of $\underline{X}$ from its quantized version.
   
Let us use a similar notation to \cite{chou1989entropy}. Denote the vector quantizer $\alpha: \underline{x} \rightarrow \mathpzc{C}$ as a mapping from an observed sample to a cluster in $\mathpzc{C}$, where $\mathpzc{C}$ is a set of $m$ clusters.  Further, let $\gamma: \mathpzc{C} \rightarrow \mathpzc{c}$ be a mapping  from a cluster to a codeword. Therefore, the composition $\alpha \circ \gamma$ is the encoder. In the same manner, the decoder is a composition $\gamma^{-1} \circ \beta$, where $\gamma^{-1}$ is the inverse mapping from a codeword to a cluster and $\beta: \mathpzc{C} \rightarrow \underline{y}$ is the reconstruction of $\underline{x}$ from its quantized version. Therefore, the Lagrangian of the optimization problem (\ref{ECVQ minimization}) is

\begin{equation}
\label{ECVQ}
L_{\lambda}(\alpha,\beta,\gamma) =\mathbb{E}\left\{D(\underline{X},\beta\left(\alpha\left(\underline{X}\right)\right)+\lambda\left|\gamma\left(\alpha\left(\underline{X}\right)\right)\right|\right\}
\end{equation}

The ECVQ objective is to find the coder $(\alpha,\beta,\gamma)$ which minimizes
this functional. In \cite{chou1989entropy}, Chou et al. suggest an iterative descent algorithm similar to the
generalized Lloyd algorithm \cite{lloyd1982least}. Their algorithm starts with an arbitrary initial coder. Then, for a fixed $\gamma$ and $\beta$ it finds a clustering $\alpha(\underline{X})$ as the minimizer of:

\begin{equation}
\label{ECVQ_a}
\alpha(\underline{X})=\argmin_{i \in \mathpzc{C}}\left\{D(\underline{X},\beta\left(i)\right)+\lambda\left|\gamma\left(i\right)\right| \right\}.
\end{equation}
Notice that for an Euclidean distortion, this problem is simply $k$-means clustering, with a ``bias" of $\lambda\left|\gamma\left(i\right)\right|$ on its objective function.

For a fixed  $\alpha$ and $\beta$, we notice that each cluster  $i \in \mathpzc{C}$ has an induced probability of occurrence $p_i$. Therefore,  the entropy encoder $\gamma$ is designed accordingly, so that $|\gamma(i)|$ is minimized. The Huffman algorithm could be incorporated into the design algorithm at this stage. However, for simplicity, allow the fiction that codewords  can have non-integer lengths, and assign 

\begin{equation}
\label{ECVQ_b}
\left|\gamma\left(i\right)\right|=-\log(p_i).
\end{equation}

Finally, for a fixed $\alpha$ and $\gamma$, the reconstruction module $\beta$ is 
\begin{equation}
\label{ECVQ_c}
\beta(i)=\argmin_{\underline{y} \in \underline{Y}} \mathbb{E} \left\{ D\left(\underline{X},\underline{y}\right) | \alpha(\underline{X})=i \right\}.
\end{equation}
For example, for an euclidean distortion measure, $\beta(i)$'s are simply the centroids of the clusters $i \in \mathpzc{C}$.

Notice that the value objective (\ref{ECVQ}), when applying each of the three steps (\ref{ECVQ_a}-\ref{ECVQ_c}), is non-increasing. Therefore, as we apply these three steps repeatedly, the ECVQ algorithm is guarenteed to converge to a local minimum. 
Moreover, notice that for an Euclidean distortion measure, step (\ref{ECVQ_a}) of the ECVQ algorithm is a variant of the $k$-means algorithm. However, the $k$-means algorithms is known to be computationally difficult to execute as the number of observed samples increases. Hence, the ECVQ algorithm is also practically limited to a relatively small number of samples.

As in previous sections, we argue that when the alphabet size is large (corresponds to low distortion), it may be better to encode the source component-wise. This means, we would like to construct a vector quantizer such that the sum marginal entropies of $\underline{Y}$ is minimal, subject to the same distortion constraint as in (\ref{R(D)}). Specifically,

\begin{equation}
\label{R(D)_ours}
\tilde{R}\left(D\right)=\min_{P(\underline{Y}|\underline{X})}\sum_{j=1}^d H_b(Y_j)\,\, s.t. \,\, \mathbb{E}\left\{D(\underline{X},\underline{Y})\right\} \leq D
\end{equation}

Notice that for a fixed distortion value, $R\left(D\right) \leq \tilde{R}\left(D\right)$ as sum of marginal entropies is bounded from below by the joint entropy. However, since encoding a source over a large alphabet may result with  a large redundancy (as discussed in previous sections), the average codeword length of the ECVQ (\ref{ECVQ minimization}) is not necessarily lower than our suggested method (and usually even much larger). 

Our suggested version of the ECVQ works as follows: we construct $\alpha$ and $\beta$ in the same manner as  ECVQ does, but replace the Huffman encoder (in $\gamma$) with our suggested relaxation to the BICA problem (Section \ref{piece-wise solution}). This means that for a fixed $\alpha, \beta$, which induce a random vector over a finite alphabet size (with a finite probability distribution), we seek for a representation which makes its components ``as statistically independent as possible". The average codeword lengths are then achieved by arithmetic encoding on each of these components.

This scheme results not only with a different codebook, but also with a different quantizer than the ECVQ. This means that a quantizer which strives to construct a random vector (over a finite alphabet) with the lowest possible average codeword length (subject to a distortion constraint) is different than our quantizer, which seeks for a random vector with a minimal sum of marginal average codeword lengths (subject to the same distortion).

Our suggested scheme proves to converge to a local minimum in the same manner that ECVQ does. That is, for a fixed $\alpha, \beta$, our suggested relaxed BICA method finds a binary representation which minimizes the sum of marginal entropies. Therefore, we can always compare the representation it achieves in the current iteration with the representation it found in the previous iteration, and choose the one which minimizes the objective. This leads to a non-increasing objective each time it is applied. Moreover, notice that we do not have to use the complicated relaxed BICA scheme and apply the simpler order permutation (Section \ref{ordering solution}). This would only result with a possible worse encoder but local convergence is still guaranteed.

To illustrate the performance of our suggested method we conduct the following experiment:
We draw $1000$ independent samples from a six dimensional bivariate Gaussian mixture.  
We apply both the ECVQ algorithm, and our suggest BICA variation of the ECVQ, on these samples. Figure \ref{fig:ECVQ} demonstrates the average codeword length we achieve for different Euclidean (mean square error) distortion levels.

\begin{figure}[ht]
\centering
\includegraphics[width = 0.5\textwidth,bb= 0 190 600 590,clip]{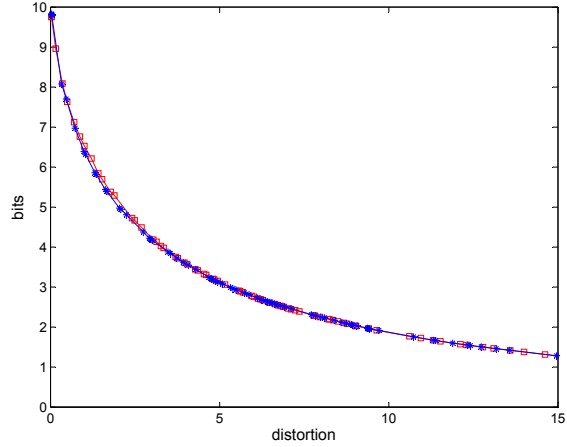}
\caption{ECVQ simulation. The curve with the squares corresponds to the average codeword length achieved by the classical ECVQ algorithm. The curve with the asterisks is the average codeword length achieved by our suggested BICA variant to the ECVQ algorithm} 
\label{fig:ECVQ}
\end{figure}

We first notice that both methods performs almost equally well. The reason is that $1000$ observations do not necessitate an alphabet size which is greater than $m=1000$ to a attain a zero distortion. In this ``small alphabet" regime, our suggest approach does not demonstrate its advantage over classical methods, as discussed in previous sections. However, we can still see it performs equally well. 

As we try to increase the number of observations (and henceforth the alphabet size) we encounter computational difficulties, which result from repeatedly performing a variant of the $k$-means algorithm (\ref{ECVQ_a}). This makes both ECVQ and our suggested method quite difficult to implement over a ``large alphabet size" (many observations and low distortion).

However, notice that if Gersho's conjecture \cite{gersho1979asymptotically} is true, and the best space-filling polytope is a lattice, then the optimum $d$-dimensional ECVQ at high resolution (low distortion) regime takes the form of a lattice \cite{zamir2014lattice}. This means that for this setup, $\gamma$ is simply a lattice quantizer. This idea is described in further detail in the next section.

\subsection{Vector Quantization with Fixed Lattices}

As demonstrated in the previous section, applying the ECVQ algorithm to a large number of observations $n$ with a low distortion constraint, is impractical. To overcome this problem we suggest using a predefined quantizer in the form of a lattice. This means that instead of seeking for a quantizer $\gamma$ that results with a random vector (over a finite alphabet) with a low average codeword length, we use a fixed quantizer, independent of the samples, and construct a codebook accordingly. Therefore, the performance of the codebook strongly depends on the empirical entropy of the quantized samples.

Since we are dealing with fixed lattices (vector quantizers), it is very likely that the empirical entropy of the quantized samples would be significantly different (lower) than the true entropy in low distortion regimes (large alphabet size). Therefore, the compressed data would consist of both the compressed samples themselves and a redundancy term, as explained in detail in Section \ref{universal source coding}.

Here again, we suggest that instead of encoding the quantized samples over a large alphabet size, we should first represent them in an ``as statistically independent  as possible" manner, and encode each component separately.

To demonstrate this scheme we turn to a classic quantizing problem, of a standard $d$-dimensional normal distribution.
Notice this quantizing problem is very well studied \cite{cover2012elements} and a lower bound for the average codeword length, for a given distortion value $D$, is given by

\begin{equation}
\label{normal R(D)}
R(D)=\max\left\{\frac{d}{2}\log\left(\frac{d}{D}\right),0\right\}.
\end{equation}
 
In this experiment we draw $n$ samples from a standard $d$-dimensional multivariate normal distribution. Since the span of the normal distribution is infinite, we use a lattice which is only defined in a finite sphere. This means that each sample which falls outside this sphere is quantized to its nearest quantization point on the surface of the sphere. We define the radius of the sphere to be $5$ times the variance of the source (hence $r=5$). We first draw $n=10^5$ samples from $d=3,4$ and $8$ dimensional normal distributions. For $d=3$ we use a standard cubic lattice, while for $d=4$ we use an hexagonal lattice \cite{zamir2014lattice}. For $d=8$ we use an $8$-dimensional integer lattice \cite{zamir2014lattice}. The upper row of Figure \ref{fig:lattice1} demonstrates the results we achieve for the three cases respectively (left to right), where for each setup we compare the empirical joint entropy of the quantized samples (dashed line) with the sum of empirical marginal entropies, following our suggested approach (solid line). We further indicate the rate distortion lower bound (\ref{normal R(D)}) for each scenario, calculated according to the true distribution (line with x's). Notice the results are normalized according to the dimension $d$.
As we can see, the sum of empirical marginal entropies is very close to the empirical joint entropy for $d=3,4$. The rate distortion indeed bounds from below both of these curves. For $d=8$ the empirical joint entropy is significantly lower than the true entropy (especially in the low distortion regime). This is a results of an alphabet size which is larger than the number of samples $n$. However, in this case too, the sum of empirical marginal entropies is close to the joint empirical entropy. The behavior described above is maintained as we increase the number of samples to $n=10^6$, as indicated in the lower row of Figure \ref{fig:lattice1}. Notice again that the sum of marginal empirical entropies is very close to the joint empirical entropy, especially on the bounds (very high and very low distortion). The reason is that in both of these cases, where the joint probability is either almost uniform (low distortion) or almost degenerate (high distortion), there exists a representation which makes the components statistically independent. In other words, both the uniform and degenerate distributions can be shown to satisfy $\sum_{j=1}^d H_b(Y_j)=H(\underline{Y})$ under the order permutation.

\begin{figure}[ht]
\centering
\includegraphics[width = 0.55\textwidth,bb= 115 105 680 500,clip]{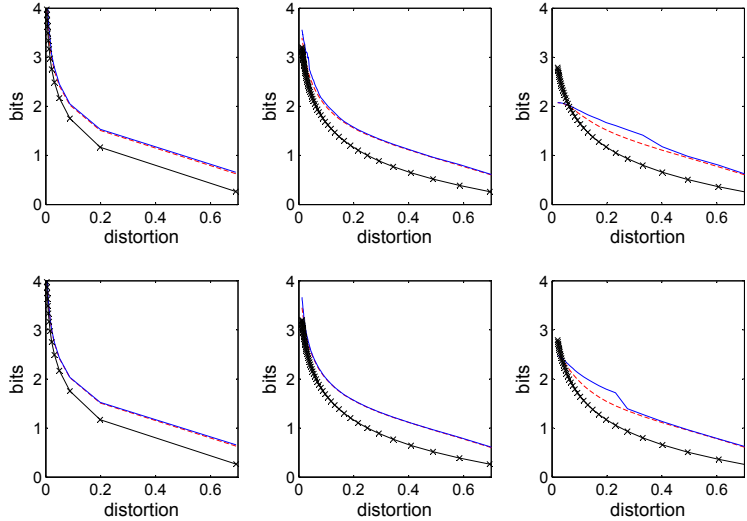}
\caption{Lattice quantization of $d$-dimensional standard normal distribution. The upper row corresponds to $n=10^5$ drawn samples while the lower row is $n=10^6$ samples. The columns correspond to the dimensions $d=3,4$ and $8$ respectively. In each setup, the dashed line is the joint empirical entropy while the solid line is the sum of marginal empirical entropies, following our suggested method. The line with the x's is the rate distortion (\ref{normal R(D)}), calculated according to the true distribution.} 
\label{fig:lattice1}
\end{figure}

We further present the total compression size of the quantized samples in this universal setting. Figure \ref{fig:lattice2} shows the amount of bits required for the quantized samples, in addition to the overhead redundancy, for both Huffman coding and our suggested scheme. As before, the rows correspond to $n=10^5$ and $n=10^6$ respectively, while the columns are $d=3,4$ and $8$, from left to right. We first notice that for $d=3,4$ both methods perform almost equally well. However, as $d$ increases, there exists a significant different between the classical coding scheme and our suggested method, for low distortion rate. The reason is that for larger dimensions, and low distortion rate, we need a very large number of quantization points, hence, a large alphabet size. This is exactly the regime where our suggested method demonstrates its enhanced capabilities, compared with standard methods.  

\begin{figure}[h]
\centering
\includegraphics[width = 0.55\textwidth,bb= 130 105 660 505,clip]{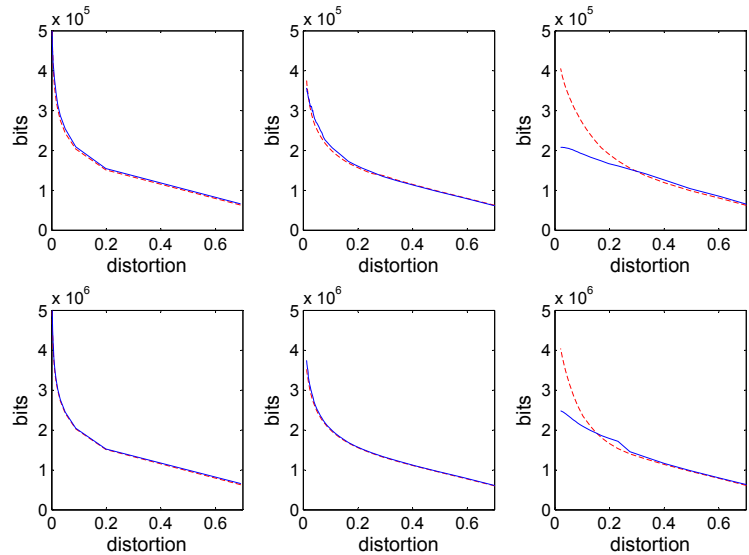}
\caption{Total compression size for lattice quantization of $d$-dimensional standard normal distribution. The upper row corresponds to $n=10^5$ drawn samples while the lower row is $n=10^6$ samples. The columns correspond to the dimensions $d=3,4$ and $8$, from left to right. In each setup, the dashed line is the total compression size through classical universal compression while the solid line is the total compression size using our suggested relaxed generalized BICA approach.} 
\label{fig:lattice2}
\end{figure}

\section{conclusions}

In this work we introduce a conceptual framework for large alphabet source coding. We suggest to decompose a large alphabet source into components which are ``as statistically independent as possible" and then encode each component separately. This way we overcome the well known difficulties of large alphabet source coding,  at the cost of:
\begin{enumerate}  [(i)]
\item  Redundancy which results from encoding each component separately.   \label{a} 
\item  Computational difficulty of finding a transformation which decomposes the source. \label{b}
\end{enumerate}

We propose two methods which focus on minimizing these costs. The first method is a piece-wise linear relaxation to the BICA (Section \ref{piece-wise solution}). This method strives to decrease (\ref{a}) as much as possible, but its computationally complexity is quite involved. Our second method is the order permutation (Section \ref{ordering solution}) which is very simple to implement (hence, focuses on  (\ref{b})) but results with a larger redundancy as it is a greedy solution to (\ref{eq:min_criterion}).

We show that while not every source can be efficiently decomposed into independent components, the vast majority of sources do decompose very well  (that is, with only a small redundancy term) as the alphabet size increases. More specifically, we show that the average difference between the sum of marginal entropies (after the ``simpler" order permutation is applied) and the joint entropy of the source is bounded by a small constant, as $m$ increases. This means that even the order permutation, which is inferior to the relaxed BICA method, is capable of achieving a very low redundancy for many large alphabet sources.

We demonstrate our suggested framework on three major large alphabet compression scenarios, which are the classic lossless source coding problem, universal source coding and vector quantization. We show that in all of these cases, our suggested approach achieves a lower average codeword length than most commonly used methods.

All this together leads us to conclude that decomposing a large alphabet source into `` as statistically independent as possible" components, followed by entropy encoding of each components separately,  is both theoretically and practically beneficial.

\section{Appendix}

\begin{appendix_proposition}
Let $\underline{X} \sim \underline{p}$ be a random vector of an alphabet size $m$ and joint probability distribution $\underline{p}$. The expected joint entropy of $\underline{X}$, when the expectation is  over a uniform simplex of joint probability distributions $\underline{p}$ is
\begin{equation}\nonumber
\mathbb{E}_{\underline{\smash{p}}}\left\{H(\underline{X}) \right\}=\frac{1}{\log_e{2}}\left(\psi(m+1)-\psi(2)\right) 
\end{equation}
where $\psi$ is the \textit{digamma function}.
\end{appendix_proposition}
\begin{proof}
We first notice that a uniform distribution over a simplex of a size $m$ is equivalent to a Direchlet distribution with parameters $\alpha_i=1, i=1, \dots,m$. The Direchlet distribution can be generated through normalized independent random variables from a Gamma distribution. This means that  for statistically independent $Z_i \sim \Gamma(k_i=1, \theta_i=1), i=1, \dots,m$ we have that 
\begin{equation}
\label{uniform simplex}
\frac{1}{\sum_{k=1}^{m}Z_k}\left(Z_1, \dots Z_m\right) \sim Dir\left(\alpha_1=1, \dots, \alpha_m=1\right).
\end{equation}

We are interested in the expected joint entropy of draws from (\ref{uniform simplex}), 
\begin{equation}
\label{joint entropy}
\mathbb{E}_{\underline{\smash{p}}}\left\{H(\underline{X}) \right\}=-\sum_{i=1}^{m}\mathbb{E}\left\{ \frac{Z_i}{\sum_{k=1}^{m}Z_k} \log{\frac{Z_i}{\sum_{k=1}^{m}Z_k}}   \right\} = -m\mathbb{E}\left\{ \frac{Z_i}{\sum_{k=1}^{m}Z_k} \log{\frac{Z_i}{\sum_{k=1}^{m}Z_k}}\right\}
\end{equation}

It can be shown that for two independent Gamma distributed random variables $X_1\sim\Gamma(\alpha_1,\theta)$ and $X_2\sim\Gamma(\alpha_2,\theta)$, the ratio $\frac{X_1}{X_1+X_2}$ follows a Beta distribution with parameters $(\alpha_1,\alpha_2)$. Let us denote $\tilde{Z}_i \triangleq \frac{Z_i}{\sum_{k=1}^{m}Z_k}=\frac{Z_i}{Z_i+\sum_{k \neq i} Z_k}$. Notice that  $Z_i\sim\Gamma(1,1)$ and $\sum_{k \neq i} Z_i\sim\Gamma(m-1,1)$ are mutually independent. Therefore,  

\begin{equation}
f_{\tilde{Z}_i}(z) =Beta(1,m-1)=\frac{(1-z)^{(m-2)}}{B(1,m-1)}.
\end{equation}

This means that
\begin{align}
\label{derivation1}
\mathbb{E}\left\{ \frac{Z_i}{\sum_{k=1}^{m}Z_k} \log{\frac{Z_i}{\sum_{k=1}^{m}Z_k}}\right\}=&\mathbb{E}\left\{ \tilde{Z}_i \log \tilde{Z}_i\right\}=\\\nonumber
&\frac{1}{B(1,m-1)}\int_0^1 z\log{(z)} (1-z)^{(m-2)}dz=\\\nonumber
&\frac{B(2,m-1)}{B(1,m-1)}\frac{1}{\log_e{(2)}}\frac{1}{B(2,m-1)}\int_0^1 \log_e{(z)}z (1-z)^{(m-2)}dz=\\\nonumber
&\frac{1}{m\log_e{(2)}}\mathbb{E}\left(\log_e{(U)}\right)
\end{align}

where $U$ follows a Beta distribution with parameters $(2, m-1)$. 

The expected natural logarithm of a Beta distributed random variable, $V\sim Beta(\alpha_1, \alpha_2)$, follows $\mathbb{E}\left(\log_e{(V)}\right)=\psi(\alpha_1)-\psi(\alpha_1+\alpha_2)$ where $\psi$ is the \textit{digamma function}.
Putting this together with (\ref{joint entropy}) and (\ref{derivation1}) we attain

\begin{equation}
\mathbb{E}_{\underline{\smash{p}}}\left\{H(\underline{X}) \right\} = -m\mathbb{E}\left\{ \frac{Z_i}{\sum_{k=1}^{m}Z_k} \log{\frac{Z_i}{\sum_{k=1}^{m}Z_k}}\right\}=\frac{1}{\log_e{(2)}}\left(\psi(m+1)-\psi(2)\right) 
\end{equation}

\begin{flushright}
$\blacksquare$
\end{flushright}

\end{proof}

\section*{Acknowledgment}
This research was supported in part by a returning scientists grant
to Amichai Painsky from the Israeli Ministry of Science, and by Israeli Science
Foundation grant 1487/12.

\bibliographystyle{IEEEtran}
\bibliography{refs}

\begin{IEEEbiographynophoto}{Amichai Painsky}
received his B.Sc. degree in Electrical Engineering from Tel Aviv University (2007) and his M.Eng. degree in Electrical Engineering from Princeton University (2009). He is currently carrying a Ph.D. at the Statistics department of Tel Aviv University School of Mathematical Sciences. His research interests include Data Mining, Machine Learning, Statistical Learning and their connection to Information Theory
\end{IEEEbiographynophoto}

\begin{IEEEbiographynophoto}{Saharon Rosset}
is an Associate Professor in the department of Statistics and Operations Research at Tel Aviv University. His research interests are in Computational Biology and Statistical Genetics, Data Mining and Statistical Learning. Prior to his tenure at Tel Aviv, he received his PhD from Stanford University in 2003 and spent four years as a Research Staff Member at IBM Research in New York. He is a five-time winner of major data mining competitions, including KDD Cup (four times) and INFORMS Data Mining Challenge, and two time winner of the best paper award at KDD (ACM SIGKDD International Conference on Knowledge Discovery and Data Mining)
\end{IEEEbiographynophoto}

\begin{IEEEbiographynophoto}{Meir Feder}
(S'81-M'87-SM'93-F'99) received the B.Sc and M.Sc degrees
from Tel-Aviv University, Israel and the Sc.D degree from the Massachusetts
Institute of Technology (MIT) Cambridge, and the Woods Hole
Oceanographic Institution, Woods Hole, MA, all in electrical
engineering in 1980, 1984 and 1987, respectively.

After being a research associate and lecturer in MIT he joined
the Department of Electrical Engineering - Systems, School of Electrical Engineering, Tel-Aviv
University, where he is now a Professor and the incumbent of the Information Theory Chair.
He had visiting appointments at the Woods Hole Oceanographic Institution, Scripps
Institute, Bell laboratories and has been a visiting
professor at MIT. He is also extensively involved in the high-tech
industry as an entrepreneur and angel investor.
He co-founded several companies including Peach Networks,
a developer of a server-based interactive TV solution which
was acquired by Microsoft, and Amimon a
provider of ASIC's for wireless high-definition A/V connectivity.

Prof. Feder is a co-recipient of the 1993 IEEE Information Theory
Best Paper Award. He also received the 1978 "creative thinking"
award of the Israeli Defense Forces, the 1994 Tel-Aviv University
prize for Excellent Young Scientists, the 1995 Research Prize of
the Israeli Electronic Industry, and the research prize in applied
electronics of the Ex-Serviceman Association, London, awarded by
Ben-Gurion University.
\end{IEEEbiographynophoto}

\vfill



\end{document}